\pdfoutput=1
\documentclass[sigconf]{acmart}

\copyrightyear{2022}
\acmYear{2022}
\setcopyright{rightsretained}
\acmConference[AFT '22]{4th ACM Conference on Advances in Financial Technologies}{September 19--21, 2022}{Cambridge, MA, USA}
\acmBooktitle{4th ACM Conference on Advances in Financial Technologies (AFT '22), September 19--21, 2022, Cambridge, MA, USA}
\acmDOI{10.1145/3558535.3559773}
\acmISBN{978-1-4503-9861-9/22/09}

\usepackage[export]{adjustbox}
\usepackage{booktabs}
\usepackage{fontawesome5}
\usepackage{csquotes}
\usepackage{mathtools, amsfonts}
\usepackage{nicefrac}
\usepackage{siunitx} 
\usepackage{subcaption}
\usepackage{tabularx}

\usepackage{tikz}
\usetikzlibrary{arrows, arrows.meta, shapes}
\tikzset{>=stealth}

\usepackage{pgfplots}
\usepgfplotslibrary{colorbrewer}
\usepgfplotslibrary{fillbetween}
\usepgfplotslibrary{statistics}
\pgfplotsset{compat=1.17}

\newcommand{\qsize}{\ensuremath{k}} %
\newcommand{\protAgree}[1][k]{\ensuremath{\mathcal{A}_{#1}}} %
\newcommand{\protChain}[1][k]{\ensuremath{\mathcal{B}_{#1}}} %

\newcommand{\hash}{\ensuremath{\mathcal{H}}} %
\newcommand{\genesisHash}{\ensuremath{H_0}} %

\def\nc{Nakamoto consensus}

\def\k{51} %
\def\citerepo{~\cite{coderepo}}
\DeclareMathOperator\opParent{parent}
\DeclareMathOperator\opQuorum{quorum}
\DeclareMathOperator\opPayload{payload}
\DeclareMathOperator\opSignature{signature}
\DeclareMathOperator\opState{state}
\DeclareMathOperator\opHeight{height}
\DeclareMathOperator\opVotes{votes}
\newcommand{\vthres}{\ensuremath{t_{\operatorname{v}}}}
\newcommand{\varReference}{\ensuremath{r}}
\newcommand{\varID}{\ensuremath{p}}
\newcommand{\varSolution}{\ensuremath{s}}
\newcommand{\initialhead}{\genesisHash} %
\newcommand{\votenp}{\ensuremath{\varReference,\varID,\varSolution}}
\newcommand{\vote}{\ensuremath{(\votenp)}}
\newcommand{\blockstore}{\ensuremath{\mathcal T}}
\providecommand\myID{\ensuremath{\operatorname{pk}}}
\providecommand\myKey{\ensuremath{\operatorname{sk}}}

\usepackage{amsthm}
\newtheorem{proposition}{Proposition}  %
\theoremstyle{definition}
\newtheorem{definition}{Definition}  %
\newtheorem{remark}{Remark}

\usepackage{algorithm} %
\usepackage[noend]{algpseudocode}

\newcommand{\event}[1]{{\normalfont$\langle$\,#1\,$\rangle$}}
\algdef{SE}[UPON]{Upon}{EndUpon}[1]{\textbf{upon event} \event{#1} \algorithmicdo}{\algorithmicend~\textbf{upon}}
\algnotext[UPON]{EndUpon}
\algdef{SE}[UPONC]{UponC}{EndUponC}[1]{\textbf{upon} #1 \algorithmicdo}{\algorithmicend~\textbf{upon}}
\algnotext[UPONC]{EndUponC}

\usepackage{etoolbox}
\makeatletter
\patchcmd{\ALG@doentity}{\item[]\nointerlineskip}{}{}{}
\makeatother
\def\algskip{\vspace{1ex}}

\begin{document}

\title{Parallel Proof-of-Work with Concrete Bounds}

\author{Patrik Keller}
\email{patrik.keller@student.uibk.ac.at}
\affiliation{
  \institution{University of Innsbruck}
  \country{Austria}
}

\author{Rainer Böhme}
\email{rainer.boehme@uibk.ac.at}
\affiliation{
  \institution{University of Innsbruck}
  \country{Austria}
}

\begin{abstract}
Authorization is challenging in distributed systems that cannot rely on the identification of nodes.
Proof-of-work offers an alternative gate-keeping mechanism, but its probabilistic nature is incompatible with conventional security definitions.
Recent related work establishes concrete bounds for the failure probability of Bitcoin's \emph{sequential} proof-of-work mechanism. %
We propose a new family of state replication protocols that use \emph{parallel} proof-of-work.
Our bottom-up design from an agreement sub-protocol allows us to give concrete bounds for the failure probability in adversarial synchronous networks.
State updates can be sufficiently secure to support commits after one block%
, removing the risk of double-spending in many applications.
We offer guidance on the optimal choice of parameters for a wide range of network and attacker assumptions.
Simulations show that the proposed construction is robust even against partial violations of our design assumptions.
\end{abstract}

\maketitle

\section{Introduction} \label{sec:intro}

Bitcoin's use of proof-of-work puzzles to secure state replication \emph{without} relying on the identification of nodes was praised as a technical novelty~\cite{abraham2017BlockchainConsensus}.
While initially supported with heuristic arguments~\cite{nakamoto2008BitcoinPeertopeer}, the security of the so-called Nakamoto consensus has  been analyzed rigorously over the past decade~\cite{garay2015BitcoinBackbone, pass2017AnalysisBlockchain, kiffer2018BetterMethod, gazi2020TightConsistency, dembo2020EverythingRace}.
All of these works prove \emph{asymptotic} security in various models.
Only recently, in AFT\,'21, Li et al.~\cite{li2021CloseLatency}
gave \emph{concrete} security bounds for the failure probability in adversarial synchronous networks.
While asymptotic bounds establish that a protocol is secure if one waits ``long enough,'' concrete bounds tell users \emph{how} long they have to wait before accepting a state update as final.
All major threats against Bitcoin's security, including double-spending and selfish mining, exploit this uncertainty in some way or another~\cite{karame2012DoublespendingFast, eyal2014MajorityNot, gervais2016SecurityPerformance, conti2018SurveySecurity}.

Nakamoto consensus uses \emph{sequential} proof-of-work: each puzzle refers back to exactly one previous puzzle solution (Fig.~\ref{fig:seq_vs_par}, left half).
A number of \emph{non-sequential} proof-of-work protocols deviate from this scheme to improve throughput or mitigate known security threats%
~\cite{kogias2016EnhancingBitcoin, sompolinsky2015SecureHighrate, sompolinsky2016SPECTREFast, bissias2020BobtailImproved}.
The approaches seem promising, but their design is heuristic.
For example, Bobtail~\cite{bissias2020BobtailImproved} argues that multiple puzzles per block imply more regular block intervals which makes it harder for attackers to double-spend.
All proposals lack security proofs, let alone concrete bounds.
Therefore, one fundamental question remains open: can non-sequential proof-of-work improve the security of state replication?

This work proposes a principled construction of state replication. %
We start from the assumptions established by the literature on sequential proof-of-work~\cite{garay2015BitcoinBackbone, pass2017AnalysisBlockchain, kiffer2018BetterMethod, gazi2020TightConsistency, dembo2020EverythingRace, li2021CloseLatency}.
We then show how agreement on the latest state can be reached with bounded worst-case failure probabilities.
By repeating the agreement procedure, we obtain a family of replication protocols that inherit the concrete error bound.
The proposed scheme uses $k$~independent puzzles per block; we thus call it $parallel$ proof-of-work (Fig.~\ref{fig:seq_vs_par}, right half).

To showcase the advantage of parallel proof-of-work, we evaluate a protocol instance that uses $k=\k$ puzzles per block while maintaining Bitcoin's expected block interval of 10 minutes.
It guarantees consistency \emph{after one block} up to a defined failure probability of \num{2.2e-4} for an attacker with 25\,\% compute power and two seconds worst-case message propagation delay (cf.~Table~\ref{tab:runtime600} below in Sect.~\ref{sec:li2021}).
Attacking consistency, e.\,g., by rewriting an already confirmed block, requires spending work on thousands of blocks without success.
For comparison, the optimal configuration of sequential proof-of-work, a ``fast Bitcoin'' with 7 blocks per minute, has a failure probability of 9\,\% in the same conditions~\cite{li2021CloseLatency}.%
\footnote{
  Bitcoin as deployed is clearly worse.
  Li et al.~\cite{li2021CloseLatency} do not even provide a failure probability for the common rule of thumb to wait six blocks (1 hour).
}
An attacker would succeed once in roughly every 2 hours.

\begin{figure}
 \centering
 \bgroup

\newcommand{\distGamma}{\ensuremath{\operatorname{Gamma}}}
\newcommand{\distExp}{\ensuremath{\operatorname{Exp}}}

\newsavebox{\seqParCheckedList}
\savebox{\seqParCheckedList}{\tikz{
    \node[] at (0,0) {\faList};
    \node[draw, circle, fill=white, inner sep=0.5pt] at (3pt,-3pt) {\tiny\faCheck};
}}

\begin{tikzpicture}[>=stealth, xscale=1.05]
  \begin{scope}[shift={(-2.5,0)}]
    \node[] at (0,0) {\usebox{\seqParCheckedList}};
    \node[] at (1,0) {\usebox{\seqParCheckedList}};
    \node[] at (2,0) {\usebox{\seqParCheckedList}};
    \node[] at (3,0) {\usebox{\seqParCheckedList}};
    \path[<-, draw] (0 + 0.3, 0) -- (1 - 0.3, 0);
    \path[<-, draw] (1 + 0.3, 0) -- (2 - 0.3, 0);
    \path[<-, draw] (2 + 0.3, 0) -- (3 - 0.3, 0);
    \draw (0 - 0.3, -0.3) rectangle (0 + 0.3, 0.3);
    \draw (1 - 0.3, -0.3) rectangle (1 + 0.3, 0.3);
    \draw (2 - 0.3, -0.3) rectangle (2 + 0.3, 0.3);
    \draw (3 - 0.3, -0.3) rectangle (3 + 0.3, 0.3);
  \end{scope}

  % \node at (-1, -1.2) {Sequential Proof-of-Work};
  % \node at ( 4, -1.2) {Parallel Proof-of-Work};
  % \node at ( 8.4, -1.2) {Legend};

  \begin{scope}[shift={(1.90,.5)}, yscale=-0.2]
    \node[inner sep=2pt] (t00) at (0,0) {\tiny\faCheck};
    \node[inner sep=2pt] (t01) at (0,1) {\tiny\faCheck};
    \node[inner sep=2pt] (t02) at (0,2) {\tiny\faCheck};
    \node[inner sep=2pt] (t03) at (0,3) {\tiny\faCheck};
    \node[inner sep=2pt] (t10) at (1,0) {\tiny\faCheck};
    \node[inner sep=2pt] (t11) at (1,1) {\tiny\faCheck};
    \node[inner sep=2pt] (t12) at (1,2) {\tiny\faCheck};
    \node[inner sep=2pt] (t13) at (1,3) {\tiny\faCheck};
    \node[inner sep=2pt] (t20) at (2,0) {\tiny\faCheck};
    \node[inner sep=2pt] (t21) at (2,1) {\tiny\faCheck};
    \node[inner sep=2pt] (t22) at (2,2) {\tiny\faCheck};
    \node[inner sep=2pt] (t23) at (2,3) {\tiny\faCheck};
    \node[inner sep=2pt] (t30) at (3,0) {\tiny\faCheck};
    \node[inner sep=2pt] (t31) at (3,1) {\tiny\faCheck};
    \node[inner sep=2pt] (t32) at (3,2) {\tiny\faCheck};
    \node[inner sep=2pt] (t33) at (3,3) {\tiny\faCheck};
    \node[inner sep=2pt] (l0) at (0,4.7) {\faList};
    \node[inner sep=2pt] (l1) at (1,4.7) {\faList};
    \node[inner sep=2pt] (l2) at (2,4.7) {\faList};
    \node[inner sep=2pt] (l3) at (3,4.7) {\faList};
    \draw (0 - 0.3,-1) rectangle (0 + 0.3,6);
    \draw (1 - 0.3,-1) rectangle (1 + 0.3,6);
    \draw (2 - 0.3,-1) rectangle (2 + 0.3,6);
    \draw (3 - 0.3,-1) rectangle (3 + 0.3,6);
    \draw[<-] (0 + 0.3, 0) -- (t10);
    \draw[<-] (0 + 0.3, 1) -- (t11);
    \draw[<-] (0 + 0.3, 2) -- (t12);
    \draw[<-] (0 + 0.3, 3) -- (t13);
    \draw[<-] (1 + 0.3, 0) -- (t20);
    \draw[<-] (1 + 0.3, 1) -- (t21);
    \draw[<-] (1 + 0.3, 2) -- (t22);
    \draw[<-] (1 + 0.3, 3) -- (t23);
    \draw[<-] (2 + 0.3, 0) -- (t30);
    \draw[<-] (2 + 0.3, 1) -- (t31);
    \draw[<-] (2 + 0.3, 2) -- (t32);
    \draw[<-] (2 + 0.3, 3) -- (t33);
  \end{scope}

  % \begin{scope}[yshift=0.5cm, yscale=0.35]
  %   \draw[|-|] (0.05, 1) -- node[scale=0.5, above] {\distGamma} (2.95, 1);
  %   \draw[|-|] (0.05, 0) -- node[scale=0.5, above] {\distExp}   (0.95, 0);
  %   \draw[|-|] (1.05, 0) -- node[scale=0.5, above] {\distExp}   (1.95, 0);
  %   \draw[|-|] (2.05, 0) -- node[scale=0.5, above] {\distExp}   (2.95, 0);
  % \end{scope}

  % \begin{scope}[yshift=-2.5cm, yscale=0.35]
  %   \draw[|-|] (0.05, 0) -- node[scale=0.5, below] {\distGamma} (0.95, 0);
  %   \draw[|-|] (1.05, 0) -- node[scale=0.5, below] {\distGamma} (1.95, 0);
  %   \draw[|-|] (2.05, 0) -- node[scale=0.5, below] {\distGamma} (2.95, 0);
  % \end{scope}
\end{tikzpicture}

\medskip
\small
\hspace{2pt}%
\tikz\node[draw,label=0:block\strut] at (0,0) {\phantom{l}};
\hfill
\tikz\node[label=0:puzzle solution\strut] at (0,0) {\scriptsize\faCheck};
\hfill
\tikz\draw[<-] (-0.17, 2) -- (0.17, 2) node[right] {hash-reference\strut};
\hfill
\tikz\node[label=0:state update\strut] at (0,0) {\faList};%
\par

\egroup
 \caption{%
   Schematic comparison of sequential (Bitcoin, left) and parallel (proposed, right) proof-of-work blockchains.
 }
 \label{fig:seq_vs_par}
\end{figure}
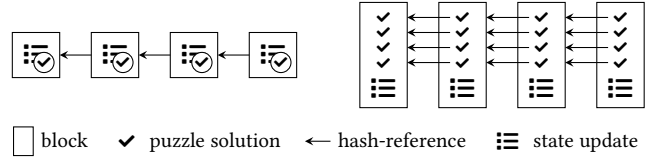

This paper makes several contributions.
We define a family of proof-of-work agreement protocols~$\protAgree$, provide upper bounds for the worst-case failure probability, and show how to find optimal parameters for different attacker and synchrony parameters.
Then we construct a family of replication protocols~$\protChain$, which invoke~$\protAgree$ iteratively to secure a blockchain.
We implement~$\protChain$ and evaluate it for robustness and security within and beyond the design assumptions using network simulation.
We parametrize our simulations to allow for a direct comparison to sequential proof-of-work as used in Bitcoin.
We offer guidance on how to parametrize~$\protChain$ for other settings.
For replicability and future research, we make the protocol and simulation code available online\citerepo{}.

The paper is organized along these contributions.
Section~\ref{sec:voting} presents and analyzes the agreement protocol.
We specify the replication protocol in Section~\ref{sec:protocol}  and evaluate it in Section~\ref{sec:eval}.
We discuss the relation to the relevant literature, limitations, and future work in Section~\ref{sec:discussion}.
Section~\ref{sec:conclusion} concludes. %

\section{Proof-of-Work Agreement} \label{sec:voting}

Agreement protocols allow their participants to unambiguously agree on a single value~\cite{pease1980ReachingAgreement}.
We consider a family of agreement protocols \protAgree{} where the participants cast votes for the value with the most existing votes until $k$ votes exist for the same value.
We rate-limit using proof-of-work: each vote requires a puzzle solution.
The votes are independent, hence puzzles can be solved in parallel.

\sloppy
In the remainder of this section, we state our model for distributed systems and proof-of-work (Sec.~\ref{sec:voting:model}),
we specify~\protAgree{} (Sec.~\ref{sec:voting:proto}), and evaluate its security considering worst-case message scheduling (Sec.~\ref{sec:voting:safe0}) as well as adversarial behavior (Sec.~\ref{sec:voting:alpha}).
We also guide the choice of parameters (Sec.~\ref{sec:voting:config}) and compare to the best known security bounds for sequential proof-of-work (Sec.~\ref{sec:voting:comparison}).

\subsection{Model} \label{sec:voting:model}

We describe an environment that simulates the execution of~\protAgree{} over continuous time.
We encode assumptions on computation, communication, and proof-of-work by explicitly defining some aspects of the environment.
In other places, we give room for worst-case behavior.
For example, we do not restrict the number of participating nodes to reflect a \enquote{permissionless} system.

\subsubsection{Event-Based Computation} \label{sec:voting:model:event}

We specify the protocol as a set of event-handlers.
The environment maintains the local state for each participating node and it invokes event-handlers for initialization, proof-of-work, and message delivery.
Each invocation takes place at a single point in time.
We write \event{ event } for events without associated data and \event{ event $\mid$ data } otherwise.

The environment invokes the \event{init~$|$~$x$} handler for each node at time~$0$.
The initialization values $x$ can be different for each node.
When a node~$A$ invokes the procedure \Call{terminate}{$x$}, the value $x$ becomes A's local result of the agreement protocol.
The environment will stop invoking further event-handlers for~$A$.
We say~$A$ terminated with return value~$x$.

\subsubsection{Communication} \label{sec:voting:model:communication}

We adopt the $\Delta$-synchronous communication model from prior analyses of \nc{}~\cite{pass2017AnalysisBlockchain, kiffer2018BetterMethod, gazi2020TightConsistency, dembo2020EverythingRace, li2021CloseLatency}.
Nodes share information by calling the \Call{broadcast}{} procedure.
If a node $A$ calls \Call{broadcast}{$m$} at time~$t$,
then the environment invokes the \event{deliver $\mid$ $m$} handler on each receiving node $B \neq A$ at or before time $t + \Delta$.
This reflects a setting where a network-level attacker can delay any message for up to~$\Delta$ time.

\subsubsection{Proof-of-Work} \label{sec:voting:model:pow}

We adopt the continuous time model for proof-of-work from prior analyses of \nc~\cite{sompolinsky2015SecureHighrate, ren2019AnalysisNakamoto, dembo2020EverythingRace, li2021CloseLatency}:
the environment has a proof-of-work mechanism $P_\lambda$ that activates nodes at random times.
In stochastic terms, $P_\lambda$ is a homogeneous Poisson process with rate $\lambda$.
For this presentation, it is instructive to describe $P_\lambda$ as a stochastic clock, where
the delays between consecutive ticks are independent and exponentially distributed with rate parameter~$\lambda$. %
We define $\bar{d} = \nicefrac1\lambda$ for the expected delay between consecutive ticks.

Let $t_i$ denote the time of the $i$-th tick.
The environment activates exactly one node per tick at the corresponding time~$t_i$ by invoking the \event{activate} handler.
We call this invocation the $i$-th activation.

In practice, proof-of-work is implemented using hash-based cryptographic puzzles that have to be solved by trial-and-error.
Attackers cannot choose \emph{which} node finds the next solution.
In our model, the environment may choose which node is activated at each tick.
Thereby, we eliminate one source of randomness and replace it with a worst-case realization.

\subsection{Protocol \protAgree} \label{sec:voting:proto}

We specify~\protAgree{} in Algorithm~\ref{alg:agree}.
During initialization, each node sets the preferred value (ln.~\ref{lv:initpref}) and initializes the vote counters to zero for all values (ln.~\ref{lv:initcount}).
Whenever a node is activated by the environment (through proof-of-work), it broadcasts a vote for its preferred value and updates the vote counter accordingly (ln.~\ref{lv:vote}-\ref{lv:incrloc}).
All nodes count the received votes and update their preference to the value with the highest counter (ln.~\ref{lv:incr}-\ref{lv:pref}).
After receiving the $k$-th vote for a value~$x$, the nodes terminate returning~$x$ (ln.~\ref{lv:term}).

\begin{algorithm}[t]
  \caption{Agreement protocol \protAgree} \label{alg:agree}
  \begin{algorithmic}[1]
    \Upon{init $|$ $x$} \label{lv:init}
    \State $p \gets x$ \Comment{preferred value} \label{lv:initpref}
    \For{$y \in \mathbb N$}
    $\operatorname{votes}(y) \gets 0$ \label{lv:initcount}
    \EndFor
    \EndUpon

    \algskip

    \Upon{activation} \Comment{proof-of-work, see Sec.~\ref{sec:voting:model:pow}}
    \State \Call{broadcast}{vote $p$} \label{alg:agree:broadcast} \label{lv:vote}
    \State $\operatorname{votes}(p) \gets \operatorname{votes}(p) + 1$ \label{lv:incrloc}
    \EndUpon

    \algskip

    \Upon{deliver $|$ vote $x$}
    \State $\operatorname{votes}(x) \gets \operatorname{votes}(x) + 1$ \label{lv:incr}
    \If{$\operatorname{votes}(x) > \operatorname{votes}(p)$}
    $p \gets x$ \label{lv:pref}
    \EndIf
    \EndUpon

    \algskip

    \UponC{$\exists x \mid \operatorname{votes}(x) \geq k$} \label{lv:termcond}
    \State \Call{terminate}{$x$} \label{lv:term}
    \EndUponC
  \end{algorithmic}
\end{algorithm}

Similar protocols exist in the literature on Byzantine Fault Tolerance (BFT).
Typically, they rely on round-based execution where each node is activated once per round.
Such rate-limiting requires strong identification of all participating nodes.
We eliminate this requirement by sourcing the activations from proof-of-work~\cite{abraham2017BlockchainConsensus}.

\subsection{Security Against Network-Level Attacks} \label{sec:voting:safe0}

Good agreement protocols ensure that all nodes terminate (liveness) and that they terminate with the same value (safety)~\cite{garay2020SoKConsensus}.
We analyze how the choice of~$k$ affects the liveness and safety of~\protAgree{}.
Our analyses depend on the parameters of the environment, i.\,e., the maximum propagation delay~$\Delta$ and the proof-of-work rate~$\lambda$.
For now, we assume that all nodes follow the protocol.
We analyze adversarial behavior in Section~\ref{sec:voting:alpha}.

Liveness is straightforward.
For $n$ nodes, there are at most $n$ different initialization values.
After $n \cdot k$ activations, at time $t_{n \cdot k} + \Delta$ more specifically, there must be one value for which each of the nodes has at least $k$ votes.
This implies termination of all nodes.

Safety is straightforward for the special case $\Delta = 0$.
Message broadcast and the corresponding message delivery happen at the same time.
At the first activation $t_1$, the activated node broadcasts a vote for its preferred value.
The other nodes receive it immediately and update their preferred value accordingly.
From then on, all nodes stay synchronized as they keep voting for the same value.
At time~$t_k$ all nodes have $k$ votes for the same value and terminate.

For $\Delta > 0$, safety becomes more involved.
The broadcast of one vote might overlap with the next activation,
two votes might cancel out, and
the synchronization can be delayed.
In order to analyze these failure modes, we first define our notion of safety.

\begin{definition}[Failure]
  We say inconsistent termination or inconsistency failure, if an execution results in two or more nodes terminating with different values.
\end{definition}

\begin{definition}[Safety] \label{def:safety}
  \protAgree{} is $\varepsilon$-safe for a given parametrization~($\Delta, \lambda$) of the environment,
  if the probability that~\protAgree{} executed in the environment results in inconsistent termination is at most~$\varepsilon$.
  Probabilities are taken over the realization of the stochastic clock~$P_\lambda$.
  Initialization values, message propagation delays, and choice of activated nodes are set to the worst case for the given realization.
\end{definition}

To show $\varepsilon$-safety, we first argue that certain realizations $\{t_i\}$ of the random activation times imply synchronization on a single preferred value.
We then measure the probability of such realizations.

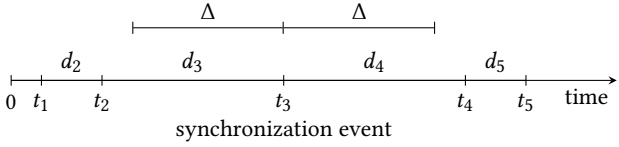
\begin{figure}
  \centering
  \def\ta{0.05}
  \def\tb{0.15}
  \def\tc{0.45}
  \def\td{0.75}
  \def\te{0.85}
  \def\dd{0.25}
  \begin{tikzpicture}[x=.95\linewidth]
    \draw[->] (0,0) -- (1,0) node[below left] {time};
    \draw (0, -2pt) node[below] {0} -- +(0, 4pt);
    \draw (\ta, -2pt) node[below] {$t_1$} -- +(0, 4pt);
    \draw (\tb, -2pt) node[below] {$t_2$} -- +(0, 4pt);
    \draw (\tc, -2pt) node[below, align=center] {$t_3$\\synchronization event} -- +(0, 4pt);
    \draw (\td, -2pt) node[below] {$t_4$} -- +(0, 4pt);
    \draw (\te, -2pt) node[below] {$t_5$} -- +(0, 4pt);
    \draw (0,0) -- %
       (\ta, 0) -- node[midway, above] {$d_2$}
       (\tb, 0) -- node[midway, above] {$d_3$}
       (\tc, 0) -- node[midway, above] {$d_4$}
       (\td, 0) -- node[midway, above] {$d_5$}
       (\te, 0);
     \draw[|-|] (\tc, .7) -- node[above, midway] {$\Delta$} +(-\dd, 0);
     \draw[ -|] (\tc, .7) -- node[above, midway] {$\Delta$} +( \dd, 0);
  \end{tikzpicture}
  \caption{%
    Activation times $t_i$ and activation delays $d_i$ for one particular realization of $P_\lambda$.
    The third activation is a synchronization event.
  }
  \label{fig:sync:timeline}
\end{figure}

\begin{definition}[Activation delay] \label{def:delay}
  Let $\{t_i\}$ be a realization of the random activation times and let $t_0 = -\infty$.
  We define $d_i = t_i - t_{i-1}$. We call $d_i$ the $i$-th activation delay.
\end{definition}

\begin{definition}[Synchronization event] \label{def:sync}
  We say that $t_i$ is a synchronization event, if both $d_i > \Delta$ and $d_{i+1} > \Delta$.
\end{definition}

Figure~\ref{fig:sync:timeline} illustrates these definitions for one realization.
Similar concepts were previously called
\emph{uniquely successful round} \cite{garay2015BitcoinBackbone},
\emph{convergence opportunity} \cite{pass2017AnalysisBlockchain}, and
\emph{loner} \cite{dembo2020EverythingRace, li2021CloseLatency}.

\begin{proposition} \label{prop:sync}
  If~$t_i$ is a synchronization event, then
  all running nodes prefer the same value at time $t_i + \Delta$.
\end{proposition}

\begin{proof}
  Let $d_i > \Delta$ and $d_{i+1} > \Delta$.
  This restriction imposes the following order of events,
  \begin{align}
    t_{i-1} < t_{i-1} + \Delta < t_i < t_i + \Delta < t_{i+1} \ .
  \end{align}
  Observe that the first $i-1$ votes are fully propagated at the time of the $i$-th activation.
  Just before the $i$-th activation, all nodes see the same votes.
  The order of votes does not matter.
  If nodes prefer different values, then there is a tie between these preferred values.

  One node is activated at time $t_i$ and votes for its preferred value~$x$.
  The other nodes receive the vote until~$t_i + \Delta$.
  Receiving nodes that prefer~$x$ leave their preference unchanged.
  Receiving nodes that prefer a different value adopt~$x$ because the new vote is breaking the tie.
  Activation $i+1$ happens later, thus there is no other vote that can interfere.
\end{proof}

\begin{proposition} \label{prop:success}
  Let $\{t_i\}$ be a realization where the
  first synchronization event happens before~$t_{2k}$.
  Then all nodes running~\protAgree{} return the same value.
\end{proposition}

\begin{proof}
  Two nodes terminating with different values requires at least~$2k$ votes (Alg.~\ref{alg:agree}, l.~\ref{lv:termcond}).
  Let $t_i$ denote the first synchronization event in $\{t_i\}$ and let $i < 2k$.
  At time~$t_i$, less than~$2k$ votes exist and all nodes are aware of all existing votes.
  If one node has terminated returning~$x$, then all nodes have terminated returning~$x$.
  Otherwise, all nodes are still running.
  By Proposition~\ref{prop:sync} all nodes prefer the same value~$y$ at time~$t_i+\Delta$.
  Nodes activated at or after $t_{i+1}$ will vote for~$y$ until all nodes terminate returning~$y$.
\end{proof}

Proposition~\ref{prop:success} provides a sufficient condition for consistency which depends on the realization of the stochastic clock.
To measure the space of realizations satisfying this condition, we construct a discrete Markov chain with three states.
The random state transitions happen at the ticks of the stochastic clock~$P_\lambda$.
Before the first synchronization event, we use two states to track whether the last delay was greater than~$\Delta$ (state ~$s_2$) or not (state~$s_1$).
The model enters the terminal state~$s_3$ if two consecutive delays are greater than~$\Delta$.
Since $d_1 = \infty$ by Definition~\ref{def:delay}, we set the start state to~$s_2$.
By construction, the Markov chain is in state~$s_3$ after $i$~transitions if and only if there was a synchronization event at or before time~$t_i$.
Table~\ref{tab:safe0} describes the states and transitions more formally.

\begin{table}
  \caption{Discrete state transitions modelling Proposition~\ref{prop:success}.}
  \label{tab:safe0}
  \setlength\tabcolsep{4pt}
  \begin{tabularx}{\columnwidth}{cl *2{>{\centering\arraybackslash}X}}
    \toprule
    \multicolumn{2}{c}{state after $i$ transitions}
                               & \multicolumn{2}{c}{state after $i \to i+1$} \\
                               \cmidrule(r){1-2}
                               \cmidrule(l){3-4}
    \# & interpretation & $d_{i+2} \leq \Delta$ & $d_{i+2} > \Delta$ \\
    \midrule
    $s_1$ & $d_{i+1} \leq \Delta~\wedge$ no synchronization event & $s_1$ & $s_2$ \\
    $s_2$ & $d_{i+1} > \Delta~\wedge$ no synchronization event    & $s_1$ & $s_3$ \\
    $s_3$ & $\exists\,j \leq i$ $\mid$ $t_j$ is synchronization event & $s_3$ & $s_3$ \\
    \bottomrule
  \end{tabularx}
\end{table}

For $i>1$, the activation delays $d_{i}$ are independent and exponentially distributed with rate~$\lambda$ by the definition of the stochastic clock (Sec.~\ref{sec:voting:model:pow}).
The probability that $d_i \leq \Delta$ is $1 - e^{-\lambda\Delta}$.
This gives us the Markov chain depicted in Figure~\ref{fig:safe0:mc}.%

\begin{figure}
  \begin{center}
    \def\p{e^{-\lambda \Delta}}
    \begin{tikzpicture}[x=.3\linewidth]
      \small
      \node[draw, circle] (s1) at (0, 0) {$s_1$};
      \node[draw, circle] (s2) at (1, 0) {$s_2$};
      \node[draw, circle] (s3) at (2, 0) {$s_3$};
      \draw[->] (s2) to[bend  left=17] node[midway, above] {$\p$}     (s3);
      \draw[->] (s2) to[bend  left=17] node[midway, below] {$1 - \p$} (s1);
      \draw[->] (s1) to[bend  left=17] node[midway, above] {$\p$}     (s2);
      \draw[->] (s3) to[out= 30,in=330,looseness=7] node[midway, right] {$1$} (s3);
      \draw[->] (s1) to[out=210,in=150,looseness=7] node[midway, left] {$1-\p$} (s1);
      \draw[<-] (s2) to +(0,0.7) node[fill=white] {start};
    \end{tikzpicture}
  \end{center}
  \caption{%
    Graphical representation of the Markov chain.
  }
  \label{fig:safe0:mc}
\end{figure}
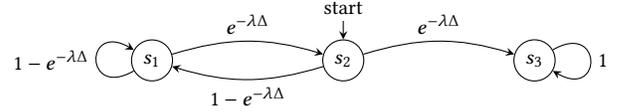

\begin{proposition} [Safety] \label{prop:safe0}
  \protAgree{} is $b_0(\lambda, \Delta, k)$-safe for
  \def\p{e^{-\lambda \Delta}}
  \begin{align}
    \textbf M(\lambda, \Delta) &=
    \begin{pmatrix}
      1 - \p & \p & 0  \\ %
      1 - \p & 0  & \p \\ %
           0 & 0  & 1  \\ %
    \end{pmatrix} ,\\
    \textbf v &=
    \begin{pmatrix}
      0 & 1 & 0 \\
    \end{pmatrix} ,\text{ and} \\
    b_0(\lambda, \Delta, k) &= 1 - \left(
      \textbf v \times \textbf M(\lambda, \Delta)^{2k - 1}
    \right)[3]\, ,
  \end{align}
  where $[3]$ denotes selection of the third element.
\end{proposition}

\begin{proof}
  $\textbf M(\lambda, \Delta)$ describes the Markov chain depicted in Figure~\ref{fig:safe0:mc} in matrix form.
  We assign vector $\textbf v$ to initialize the Markov model in state $s_2$.
  The third element of the result of $\textbf v \times \textbf M(\lambda, \Delta)^{2k - 1}$ describes the probability that the Markov model is in state $s_3$ after $2k - 1$ random transitions.
  Our claim follows by the construction of the Markov chain, Definition~\ref{def:safety}, and Proposition~\ref{prop:success}.
\end{proof}

Observe that $b_0$ depends only on the vote threshold~$k$ and the product of $\lambda$ and $\Delta$.
It is instructive to interpret $\lambda$ as the inverse of the expected activation delay~$\bar{d}$ and use the synchrony parameter~$\Delta$ as a unit of time.
Figure~\ref{fig:safe0:bound} visualizes the bound $b_0(\lambda, \Delta, k)$ for different combinations of $\bar{d}$ (as multiple $\Delta$) on the x-axis and $k$ on the y-axis.
Both parameters have a positive effect on the safety of \protAgree{}.

Proposition~\ref{prop:safe0} guarantees $\varepsilon$-safety for~\protAgree{} against strong network-level attacks.
We now extend the argument to adversarial voting.

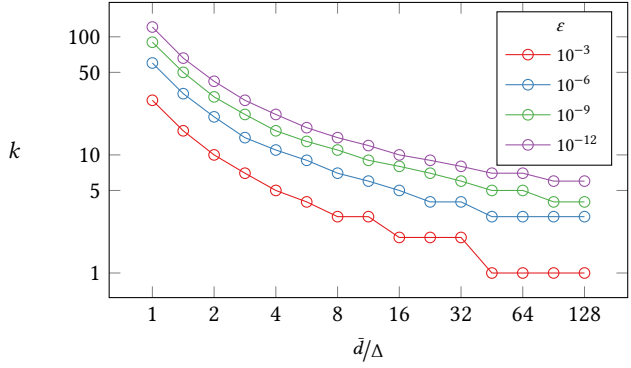
\begin{figure}
  \begin{tikzpicture}
    \def\data{voting-safe0-bound.csv}
    \def\xcol{x}
    \begin{axis}[
      table/col sep=comma,
      width=\linewidth,
      height=0.65\linewidth,
      xmode=log, ymode=log,
      xtick={1,2,4,8,16,32,64,128},
      xticklabels={1,2,4,8,16,32,64,128},
      ytick={1, 5, 10, 50, 100},
      yticklabels={1, 5, 10, 50, 100},
      xlabel={$\nicefrac{\bar{d}}{\Delta}$},
      ylabel={$k$},
      y label style={at={(axis description cs:-0.21,.5)},rotate=-90,anchor=west},
      legend cell align=left,
      legend pos= north east,
      legend style={font=\footnotesize},
      cycle list/Set1-4,
      cycle multiindex* list={ mark=o\nextlist Set1-4\nextlist }
      ]
      \addlegendimage{empty legend}
      \addlegendentry{$\varepsilon$};
      \foreach \eps in {3, 6, 9, 12} {
        \addplot table[y=k.\eps, x=\xcol] {\data};
        \addlegendentryexpanded{$10^{-\eps}$};
      }
    \end{axis}
  \end{tikzpicture}
  \caption{Minimum $k$, such that \protAgree{} is $\varepsilon$-safe by Proposition~\ref{prop:safe0}.}
  \label{fig:safe0:bound}
\end{figure}

\subsection{Security Against Malicious Voting} \label{sec:voting:alpha}

We now consider attackers who can cast votes.
While equivocation is a major concern in the setting with identified nodes, the use of proof-of-work in $\protAgree{}$ completely removes this issue~\cite{abraham2017BlockchainConsensus}:
in practice, every vote is authenticated with a costly solution to a computational puzzle~\cite{dwork1993PricingProcessing}.
Therefore, the only remaining attack strategy is to withhold votes to release them later.
We address this next.

\subsubsection{Attacker Votes}

We assume that the attacker controls at most~$\alpha$ of the total proof-of-work capacity.%
\footnote{The analysis generalizes to multiple attackers: we conservatively assume that all attackers coordinate perfectly and sum up their shares.}
We model this by splitting the stochastic clock $P_{\lambda}$ in two independent parts~$P_A$ and~$P_D$.
The clock~$P_A$ ticks at rate $A = \alpha \cdot \lambda$.
With each tick, the attacker gains one \emph{attacker vote} which can be withheld arbitrarily before being sent.
The clock~$P_D$ ticks at rate $D = (1-\alpha) \cdot \lambda$ and results in activations of defender nodes which follow the protocol as specified.

Recall that the sum of multiple Poisson processes is another Poisson process with cumulated rate.
The probability that an individual tick of $P_{\lambda} = P_{A + D}$ results in an attacker vote is $\alpha$.

\subsubsection{Malicious Voting}
We study withholding attacks by considering two phases.
In the \textbf{balancing} phase nodes are not yet synchronized on the same preferred value.
Recall that without malicious voting, the nodes would synchronize at the first synchronization event (see Prop.~\ref{prop:sync}).
A vote-withholding attacker can prevent synchronization by releasing withheld votes around synchronization events.
The balancing phase continues while the attacker can balance synchronization events with his stock of withheld votes.
If the attacker does not release a withheld vote around a synchronization event, e.\,g.~because his stock is empty, the nodes synchronize.
This is when the attack transitions to the \emph{catching-up} phase.

During the \textbf{catching-up} phase, all nodes prefer the same value.
With each tick of the stochastic clock~$P_D$, the nodes cast one vote for this value and thereby reinforce the synchronization.
However, the attacker can destroy the synchronization by releasing sufficiently many votes for a different value.
If this happens, the attack transitions back to the \emph{balancing} phase.

Both phases can be characterized with an integer depth.
In the \emph{balancing} phase it matters how many votes are currently withheld by the attacker.
In the \emph{catching-up} phase it matters how many votes the attacker has to cast in order to destroy the synchronization.
Our attack model tracks these depths in a single \textbf{margin} variable $m$.
Positive $m$ represent withheld votes during \emph{balancing} and negative~$m$ represent number of votes to be \emph{caught-up} (see Fig.~\ref{fig:margin}).

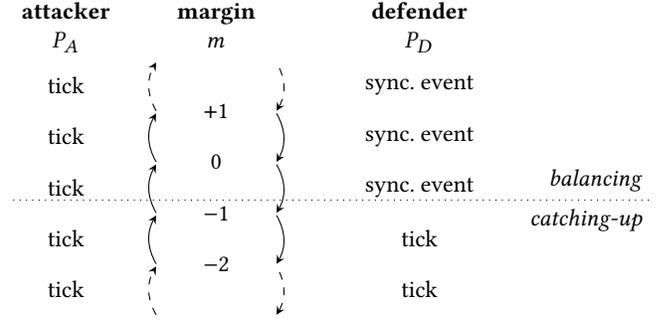
\begin{figure}
  \def\l{-2}
  \def\r{2.7}
  \begin{tikzpicture}[y=5ex]
    \node[align=center, inner sep=0, outer sep=0] at (\r, 2.6) {\textbf{defender}\\$P_D$\strut};
    \node[align=center, inner sep=0, outer sep=0] at (\l, 2.6) {\textbf{attacker}\\$P_A$\strut};
    \node[align=center, inner sep=0, outer sep=0] at ( 0, 2.6) {\textbf{margin}\\$m$\strut};
    \foreach \m in {+1, 0, -1, -2} \node at (0, \m) {$\m$};
    \foreach \m in {1, -3} {
      \draw[->, dashed] (-0.8, \m) to[bend left] +(0, 0.95);
      \draw[<-, dashed] (0.8, \m) to[bend right] +(0, 0.95);
    }
    \foreach \m in {0, -1, -2} {
      \draw[->] (-0.8, \m) to[bend left] +(0, 0.95);
      \draw[<-] (0.8, \m) to[bend right] +(0, 0.95);
    }
    \foreach \m in {1, 0, -1, -2, -3} \node[] at ([shift={(0, .5)}] \l, \m) {tick};
    \foreach \m in {2, 1, 0} \node[] at ([shift={(0, -.5)}] \r, \m) {sync.~event};
    \foreach \m in {-1, -2} \node[] at ([shift={(0, -.5)}] \r, \m) {tick};
    \draw[dotted] (-2.7, -0.75) -- +(\linewidth, 0)
      node[above left] {\emph{balancing}}
      node[below left] {\emph{catching-up}};
  \end{tikzpicture}
  \caption{Security of $\protAgree$: attacker's margin and the transition between \emph{balancing} and \emph{catching-up} phases.}
  \label{fig:margin}
\end{figure}

\subsubsection{Markov Chain Model}

We quantify the safety of~\protAgree{} against malicious voting by measuring the space of realizations of the joint clock $P_\lambda$ where withholding and propagation delays enable inconsistent termination.
We generalize the Markov chain model of Section~\ref{sec:voting:safe0} to include the states for different margins $m$.
The new state space is $(m,s) \in \mathbb{Z} \times \{\bot, \top\}$.
To track the occurrence of synchronization events, we set $s=\top$ if and only if the last activation delay was greater~$\Delta$.
The initial state is $(0, \top)$: zero votes are withheld and $s=\top$ since $d_1 > \Delta$ by Definition~\ref{def:delay}.

As before, transitions happen at each tick.
For attacker votes (probability $\alpha$), we increment $m$.
Depending on the phase, the attacker withholds (increasing the stock of withheld votes) or catches up by one; both map to the same transition.
For defender activations (probability $1 - \alpha$), we distinguish the two phases.
In the \emph{balancing} phase ($m \geq 0$), we only decrement $m$ if the tick is a synchronization event.
In the \emph{catching-up} phase ($m < 0$), we always decrement $m$.
Figure~\ref{fig:safe:mc} illustrates the state transitions and transition probabilities.

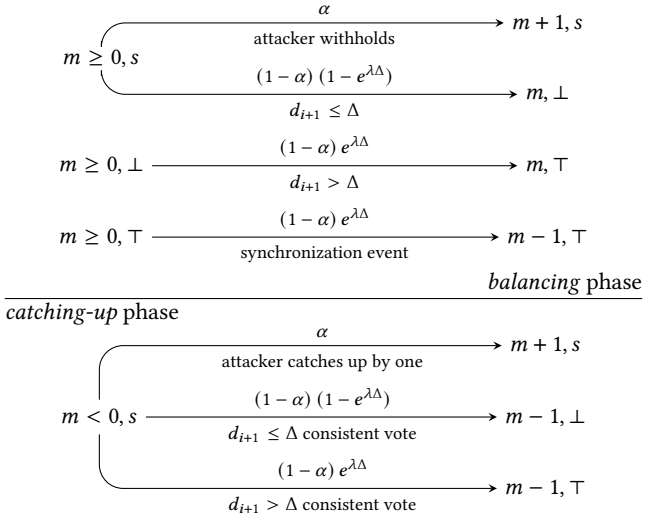
\begin{figure}
  \centering
  \tikzstyle{src}=[]
  \tikzstyle{dst}=[]
  \tikzstyle{transition}=[->, draw, rounded corners=1em]
  \tikzstyle{prob}=[anchor=south, font=\footnotesize, xshift=-0.35\linewidth]
  \tikzstyle{intuition}=[anchor=north, font=\footnotesize, xshift=-0.35\linewidth]
  \def\p{e^{\lambda\Delta}}
  \def\r{(\alpha + \lambda)}
  \begin{tikzpicture}[x=0.7\linewidth,y=-7ex]
    \node[src] (a) at (0,.5) {$m \geq 0, s$};
    \node[dst] (aa) at (1,0) {$m+1, s$};
    \node[dst] (ab) at (1,1) {$m, \bot$};
    \path[transition] (a) |- (aa);
    \path[transition] (a) |- (ab);
    \node[prob] at (aa) {$\alpha$};
    \node[prob] at (ab) {$(1-\alpha)\,(1 - \p)$};
    \node[intuition] at (aa) {attacker withholds};
    \node[intuition] at (ab) {$d_{i+1} \leq \Delta$};

    \node[src] (b) at (0,2) {$m \geq 0, \bot$};
    \node[dst] (ba) at (1,2) {$m, \top$};
    \path[transition] (b) -- (ba);
    \node[prob] at (ba) {$(1-\alpha)\,\p$};
    \node[intuition] at (ba) {$d_{i+1} > \Delta$};

    \node[src] (c) at (0,3) {$m \geq 0, \top$};
    \node[dst] (ca) at (1,3) {$m-1, \top$};
    \path[transition] (c) -- (ca);
    \node[prob] at (ca) {$(1-\alpha)\,\p$};
    \node[intuition] at (ca) {synchronization event};
  \end{tikzpicture}

  \vspace{-1ex}
  \strut\hfill\emph{balancing} phase
  \hrule
  \emph{catching-up} phase \hfill\strut
  \vspace{-1ex}

  \begin{tikzpicture}[x=0.7\linewidth,y=-7ex]
    \node[src] (d) at (0,5) {$m < 0, s$};
    \node[dst] (da) at (1,4) {$m + 1, s$};
    \node[dst] (db) at (1,5) {$m - 1, \bot$};
    \node[dst] (dc) at (1,6) {$m - 1, \top$};
    \path[transition] (d) |- (da);
    \path[transition] (d) -- (db);
    \path[transition] (d) |- (dc);
    \node[prob] at (da) {$\alpha$};
    \node[prob] at (db) {$(1-\alpha)\,(1-\p)$};
    \node[prob] at (dc) {$(1-\alpha)\,\p$};
    \node[intuition] at (da) {attacker catches up by one};
    \node[intuition] at (db) {$d_{i+1} \leq \Delta$ ~~consistent vote};
    \node[intuition] at (dc) {$d_{i+1} > \Delta$ ~~consistent vote};
  \end{tikzpicture}
  \caption{%
    Generalized Markov chain model for attackers who can send and withhold votes.
    Transition probabilities are annotated above the arrows and interpretations below.
  }
  \label{fig:safe:mc}
\end{figure}

\subsubsection{Numerical Solution}

In principle, within $l$ steps, the model can reach any state with $-l \leq m \leq l$.
Calculating the exact state probabilities after $l$ transitions requires us to raise a square matrix with $n = 2(2l +1)$ rows to the power of $l$.
Each matrix multiplication is $O\left(n^{2.8}\right)$~\cite{strassen1969GaussianElimination}.
Thus, the analysis is infeasible for larger~$l$.

We set a cut-off at $m = \pm 25$ to make the problem tractable.
We assume that an attacker who manages to withhold 25 votes during the \emph{balancing} causes inconsistent termination.
Similarly, an attacker lagging behind 25 votes in the \emph{catching-up} phase cannot catch up at all.
With these assumptions, the number of states is bounded by $102$ and the matrix multiplications stay tractable.
Using such cut-offs is common practice in the related literature~\cite{gervais2016SecurityPerformance, sapirshtein2016OptimalSelfish}. %

A second simplification in this model is that it does not track how many votes are cast for each value.
Adding this information would blow up the state space excessively.
We work around this problem by ignoring the termination rule of \protAgree{} and assume that the nodes continue voting forever.
We thus need to rephrase our notions of success and failure for the purpose of this analysis.

Recall that inconsistent commits require at least $2k$ votes. %
We count an execution as successful if all nodes prefer the same value after $2k$ steps.
This is easy to check by inspecting the phase after $2k$ transitions: \emph{catching-up} means success and \emph{balancing} means failure.

We calculate the failure probability of \protAgree{} for different combinations of $\alpha$, $\Delta$, $\lambda = \nicefrac{1}{\bar{d}}$, and $k$ by exponentiation of the probability matrix of the generalized Markov chain model.
We visualize this in Figure~\ref{fig:safe:bound:k}, following the setup of Figure~\ref{fig:safe0:bound}, but with more lines for different assumptions of attacker strength $\alpha$.
As expected, increasing $\alpha$ pushes up the required $k$ for any given failure bound $\varepsilon$ and expected activation delay $\bar{d}$.
For example, assuming a proof-of-work puzzle takes 8 times the maximum propagation delay, while without attacker, $k=3$ were sufficient for $10^{-3}$-safety, $k$ must increase to $k=9$ if an attacker is present and controls 10\,\% of the proof-of-work capacity; or to $k=88$ for 33\,\% attacker strength.

In practice, a protocol designer can adjust the puzzle difficulty and should care about the protocol runtime, to which we turn next.

\begin{figure}
  \begin{tikzpicture}
    \pgfplotsset{set layers}
    \def\data{mc-py/bounded-failure-wide.csv}
    \begin{axis}[
      table/col sep=comma,
      unbounded coords=jump,
      width=\linewidth,
      height=0.65\linewidth,
      xmode=log, ymode=log,
      xtick={1,2,4,8,16,32,64,128},
      xticklabels={1,2,4,8,16,32,64,128},
      ytick={1, 2, 5, 10, 20, 50, 100, 200, 500},
      yticklabels={1, 2, 5, 10, 20, 50, 100, 200, 500},
      xlabel={$\nicefrac{\bar{d}}{\Delta}$},
      ylabel={$k$},
      y label style={at={(axis description cs:-0.21,.5)},rotate=-90,anchor=west},
      cycle list/Set1-3, %
      cycle multi list={ Set1-3\nextlist mark=x,mark=triangle,mark=star,mark=o\nextlist},
      cycle list shift=-11 %
      ]
      \foreach \atk in {10,25,33} {
        \foreach \eps in {1,2,3,4} {
          \addplot table[y=k.\atk.\eps, x=interval] {\data};
          \label{curve.\atk.\eps}
        }
      }
      \begin{pgfonlayer}{axis foreground}
        \node [draw,fill=white, anchor=north east, font=\footnotesize] at (axis description cs: .97, .97) {
            \shortstack[c]{
              $\varepsilon$\\[2pt]
              \ref{curve.25.1} $10^{-1}$ \\
              \ref{curve.25.2} $10^{-2}$ \\
              \ref{curve.25.3} $10^{-3}$ \\
              \ref{curve.25.4} $10^{-4}$
          }};
        \node [draw,fill=white, anchor=south west, font=\footnotesize] at (axis description cs: .03,.03) {
            \shortstack[c]{
              $\alpha$\\[2pt]
              \ref{curve.10.1} $10\,\%$ \\
              \ref{curve.25.1} $25\,\%$ \\
              \ref{curve.33.1} $33\,\%$
          }};
      \end{pgfonlayer}
    \end{axis}
  \end{tikzpicture}
  \caption{%
    Minimal $k$ such that \protAgree{} satisfies the given failure probability bound $\varepsilon$ for a given attacker and expected activation delay $\bar{d}$ as multiple of $\Delta$.
  }
  \label{fig:safe:bound:k}
\end{figure}
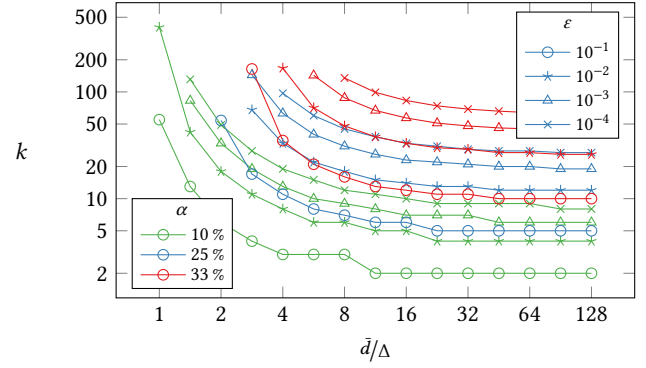

\subsection{Choosing Efficient Parameters} \label{sec:voting:config}

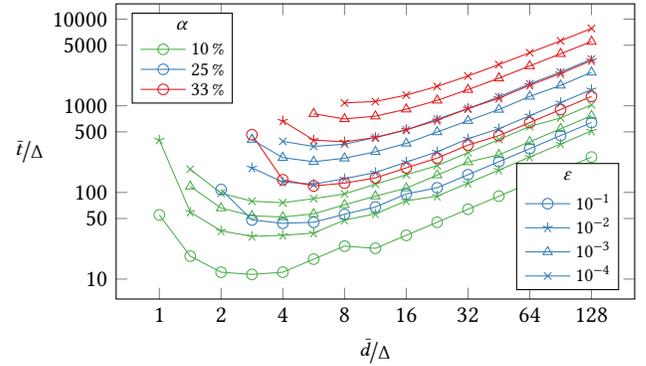
\begin{figure}
  \begin{tikzpicture}
    \def\data{mc-py/bounded-failure-wide.csv}
    \pgfplotsset{set layers}
    \begin{axis}[
      table/col sep=comma,
      unbounded coords=jump,
      width=\linewidth,
      height=0.65\linewidth,
      xmode=log, ymode=log,
      xtick={1,2,4,8,16,32,64,128},
      xticklabels={1,2,4,8,16,32,64,128},
      ytick={10, 50, 100, 500, 1000, 5000, 10000},
      yticklabels={10, 50, 100, 500, 1000, 5000, 10000},
      xlabel={$\nicefrac{\bar{d}}{\Delta}$},
      ylabel={$\nicefrac{\bar{t}}{\Delta}$},
      y label style={at={(axis description cs:-0.21,.5)},rotate=-90,anchor=west},
      cycle list/Set1-3, %
      cycle multi list={ Set1-3\nextlist mark=x,mark=triangle,mark=star,mark=o\nextlist},
      cycle list shift=-11 %
      ]
      \foreach \atk in {10,25,33} {
        \foreach \eps in {1,2,3,4} {
          \addplot table[y=runtime.\atk.\eps, x=interval] {\data};
          \label{curve0.\atk.\eps}
        }
      }
      \begin{pgfonlayer}{axis foreground}
        \node [draw,fill=white, anchor=south east, font=\footnotesize] at (axis description cs: .97, .03) {
            \shortstack[c]{
              $\varepsilon$\\[2pt]
              \ref{curve0.25.1} $10^{-1}$ \\
              \ref{curve0.25.2} $10^{-2}$ \\
              \ref{curve0.25.3} $10^{-3}$ \\
              \ref{curve0.25.4} $10^{-4}$
          }};
        \node [draw,fill=white, anchor=north west, font=\footnotesize] at (axis description cs: .03,.97) {
            \shortstack[c]{
              $\alpha$\\[2pt]
              \ref{curve0.10.1} $10\,\%$ \\
              \ref{curve0.25.1} $25\,\%$ \\
              \ref{curve0.33.1} $33\,\%$
          }};
      \end{pgfonlayer}
    \end{axis}
  \end{tikzpicture}
  \caption{%
    Protocol runtime after choosing the minimal $k$ such that \protAgree{} satisfies the given failure probability bound $\varepsilon$ for a given attacker and expected activation delay $\bar{d}$. Both axes show times as multiples of $\Delta$.
  }
  \label{fig:safe:bound:time}
\end{figure}

The aim here is to guide the choice of protocol parameters to minimize the protocol runtime for given assumptions about the real world.
\protAgree's failure probability depends on the synchrony parameter~$\Delta$, the proof-of-work rate~$\lambda$, the attackers compute power~$\alpha$, and the threshold~$k$.
The protocol operator can choose $\lambda$ and $k$, while $\Delta$ and $\alpha$ are worst-case assumptions.
Safety increases with $k$ or by decreasing $\lambda$.
But both options slow down termination:
either we wait for more votes or we wait longer for each vote.

Recall that the protocol runtime is stochastic.
Termination requires $k$ votes for the same value and thus at least $k$ activations.
The time of the $k$-th activation is the sum of $k$ exponentially distributed delays, i.\,e., gamma distributed with shape parameter~$k$.
Whenever nodes vote for different values---due to network delays or withholding---termination requires more activations and the shape parameter of the gamma distribution increases.

We optimize the protocol runtime for the optimistic case where $k$ activations enable termination.
We call $\bar{t} = k \cdot \bar{d} = k/\lambda$ the optimistic expected protocol runtime.
Figure~\ref{fig:safe:bound:time} shows~$\bar{t}$ (in multiples of $\Delta$) for the same parameters as used in Figure~\ref{fig:safe:bound:k}.
Observe that depending on $\varepsilon$ and $\alpha$, different activation delays $\bar{d}$ minimize the protocol runtime.

As the curves in Figure~\ref{fig:safe:bound:time} are neither convex nor continuous and expensive to evaluate,
we identify the minima using Bayesian optimization~\cite{bayesOpt} and report them in Table~\ref{tab:optima}.

\begin{table}
  \caption{%
    Configurations of~\protAgree{} minimizing protocol runtime~$\bar{t}$.
  }
  \label{tab:optima}
  \setlength\tabcolsep{5pt}
  \begin{tabular}{ccccc}
  \toprule
  $\alpha$ & $\varepsilon$ & $\nicefrac{\bar{d}}{\Delta}$ & $k$ & $\nicefrac{\bar{t}}{\Delta}$ \\
  \midrule
   $0$ & ${10}^{-1}$ & $3.0$ & $2$ & $6$ \\
   $0$ & ${10}^{-2}$ & $2.6$ & $5$ & $13$ \\
   $0$ & ${10}^{-3}$ & $2.5$ & $8$ & $20$ \\
   $0$ & ${10}^{-4}$ & $2.7$ & $10$ & $27$ \\
   $\nicefrac{1}{10}$ & ${10}^{-1}$ & $2.8$ & $4$ & $11$ \\
   $\nicefrac{1}{10}$ & ${10}^{-2}$ & $3.3$ & $9$ & $29$ \\
   $\nicefrac{1}{10}$ & ${10}^{-3}$ & $3.2$ & $16$ & $51$ \\
   $\nicefrac{1}{10}$ & ${10}^{-4}$ & $3.8$ & $20$ & $75$ \\
  \bottomrule
\end{tabular}%
\hfill%
\begin{tabular}{ccccc}
  \toprule
  $\alpha$ & $\varepsilon$ & $\nicefrac{\bar{d}}{\Delta}$ & $k$ & $\nicefrac{\bar{t}}{\Delta}$ \\
  \midrule
   $\nicefrac{1}{4}$ & ${10}^{-1}$ & $4.0$ & $10$ & $40$ \\
   $\nicefrac{1}{4}$ & ${10}^{-2}$ & $5.1$ & $24$ & $123$ \\
   $\nicefrac{1}{4}$ & ${10}^{-3}$ & $5.7$ & $40$ & $226$ \\
   $\nicefrac{1}{4}$ & ${10}^{-4}$ & $5.8$ & $58$ & $339$ \\
   $\nicefrac{1}{3}$ & ${10}^{-1}$ & $6.1$ & $19$ & $115$ \\
   $\nicefrac{1}{3}$ & ${10}^{-2}$ & $7.4$ & $51$ & $375$ \\
   $\nicefrac{1}{3}$ & ${10}^{-3}$ & $7.9$ & $88$ & $699$ \\
   $\nicefrac{1}{3}$ & ${10}^{-4}$ & $8.8$ & $121$ & $1067$ \\
  \bottomrule
\end{tabular}%

  \medskip
  \begin{minipage}{\linewidth}
    \small
    We set attacker $\alpha$ and failure bound $\varepsilon$.
    We optimize activation delay $\bar{d}$ and threshold $k$.
    The synchrony parameter $\Delta$ serves as unit of time.
  \end{minipage}
\end{table}

Now we see that in the above example ($\varepsilon=10^{-3}$, $\alpha=\nicefrac1{10}$), the protocol runtime can be reduced from $\bar{t}=72\,\Delta$ to $51\,\Delta$ without sacrificing safety by choosing $(k, \bar{d})=(16, 3.2\,\Delta)$ instead of $(9, 8\,\Delta)$.
For perspective, with a network latency bound $\Delta = 2$ seconds, the puzzle difficulty should be adjusted to one solution every $6.4$ seconds.
The protocol $\protAgree[16]$ would terminate in about 102 seconds.

\subsection{Comparison to Sequential Proof-of-Work}
\label{sec:voting:comparison} \label{sec:li2021}

In their contribution to AFT\,'21, Li et al.~\cite{li2021CloseLatency} provide concrete bounds for the failure probability of Nakamoto consensus.
The \enquote{achievable security latency function} $\bar{\epsilon}(t)$ for Nakamoto consensus as stated in \cite[Theorem 3.5]{li2021CloseLatency} provides an upper bound for the failure probability after waiting for a given confirmation time $t$.
Our models and assumptions are compatible, but we derive failure probabilities after termination and hence after stochastic runtime.

To enable comparison of sequential and parallel proof-of-work, we fix the time frame to $\bar{t} = 10$ minutes.
We also fix the attacker $\alpha$ and propagation delay $\Delta$.
For parallel proof-of-work, we optimize~$k$ for minimal failure probability of \protAgree{} subject to $\bar{d} \cdot k = \bar{t}$.
The resulting configuration exhibits an expected protocol runtime of 10 minutes.
For sequential proof-of-work, we optimize the block interval $\bar{d}_\text{seq}$ for minimal failure probability after 10 minutes.

Table~\ref{tab:runtime600} compares the failure probability $\varepsilon$ of \protAgree{} with the achievable security $\varepsilon_\text{seq}$ of Nakamoto consensus for various combinations of $\Delta$ and $\alpha$.
Note that Li et al.~\cite{li2021CloseLatency} do not define $\bar{\epsilon}(t)$ for all combinations of $\alpha$ and $\Delta$.
We omit the undefined values from the table.
Observe that our concrete bounds for parallel proof-of-work consistently outperform the bounds for sequential proof-of-work by at least two orders of magnitude.

\begin{remark}
  The true advantage is smaller because Li et al.~\cite{li2021CloseLatency} consider a long-running blockchain protocol, whereas our agreement protocol has a limited time horizon.
  Moreover, concurrent work presents improved bounds for sequential proof-of-work~\cite{gazi2021PracticalSettlement, guo2022BitcoinLatency}.
\end{remark}

\begin{table}
  \caption{
    Advantage of parallel over sequential proof-of-work.
  } %
  \label{tab:runtime600}
  \centering
  \setlength{\tabcolsep}{5.45pt}
  \begin{tabular}{cccS[table-format=2.1]S[table-format=1.1e2]S[table-format=3.1]S[table-format=2.1]S[table-format=1.1e2]}
    \toprule
    \multicolumn{2}{c}{Parameters} &
    \multicolumn{3}{c}{Parallel} &
    \multicolumn{3}{c}{Sequential~\cite{li2021CloseLatency}}\\
    \cmidrule(lr){1-2}
    \cmidrule(lr){3-5}
    \cmidrule(lr){6-8}
    $\Delta$ & $\alpha$ &
    $k$ & {$\bar{d}$} & {$\varepsilon$} &
    {$\bar{k}_{\text{seq}}$} & {$\bar{d}_{\text{seq}}$} & {$\varepsilon_{\text{seq}}$} \\
    \midrule
    $1$ & $\nicefrac{1}{10}$ & $77$ & 7.8  & 6.3e-20 & 192.7 & 3.1  & 9.8e-15 \\
    $1$ & $\nicefrac{1}{4}$  & $95$ & 6.3  & 7.3e-07 & 136.6 & 4.4  & 1.8e-03 \\
    $1$ & $\nicefrac{1}{3}$  & $76$ & 7.9  & 1.9e-03 & 103.5 & 5.8  & 2.6e-01 \\
    $2$ & $\nicefrac{1}{10}$ & $76$ & 7.9  & 3.9e-13 & 96.9  & 6.2  & 3.4e-07 \\
    $2$ & $\nicefrac{1}{4}$  & $51$ & 11.8 & 2.2e-04 & 68.8  & 8.7  & 8.8e-02 \\
    $2$ & $\nicefrac{1}{3}$  & $43$ & 14.0 & 1.8e-02 & {--}  & {--} & {--}    \\
    $4$ & $\nicefrac{1}{10}$ & $39$ & 15.4 & 1.2e-07 & 49.0  & 12.2 & 1.8e-03 \\
    $4$ & $\nicefrac{1}{4}$  & $28$ & 21.4 & 5.3e-03 & 34.8  & 17.2 & 5.2e-01 \\
    $4$ & $\nicefrac{1}{3}$  & $24$ & 25.0 & 6.9e-02 & {--}  & {--} & {--}    \\
    \bottomrule
  \end{tabular}

  \medskip
  \small
  \setlength{\tabcolsep}{3pt}
  \begin{tabular}{ll}
    $\Delta$ & propagation delay, seconds \\
    $\alpha$  & attacker's compute \\
    $k$ & number of votes \\
    $\bar{d}$ & activation delay, seconds \\
  \end{tabular}\hfill%
  \begin{tabular}{ll}
    $\varepsilon$  & failure probability of \protAgree{}\\
    $\bar{d}_{\text{seq}}$ & block interval, seconds \\
    $\bar{k}_{\text{seq}}$ & expected number of blocks\\
    $\varepsilon_{\text{seq}}$ & failure prob.~after 10 minutes\\
  \end{tabular}
\end{table}

\section{Proof-of-Work Blockchain} \label{sec:protocol} \label{sec:c1btc}

Protocol~\protAgree{} solves agreement using parallel proof-of-work.
In this section, we propose a replication protocol~\protChain{} that repeatedly runs~\protAgree{} to continuously agree on a growing sequence of values.
In a nutshell, \protChain{}~is a blockchain protocol where the participants use~\protAgree{} to agree on each appended block. %

Time is divided in epochs of variable length, determined by the runtime of~\protAgree{}.
Each epoch extends the blockchain by one block and confirms the value of the preceding epoch's block using~\protAgree{}.
The safety guarantees of~\protAgree{} imply that possible conflicting block proposals for the current epoch are resolved in the next epoch.

\subsection{Prerequisites} \label{sec:prot:req}

In addition to the network assumptions of \protAgree{} (Sect.~\ref{sec:voting:model:communication}), we assume interfaces to an application layer and the availability of cryptographic primitives.

\subsubsection{Application} \label{sec:proto:app}

\protChain{} enables state replication and may serve as a basis for different applications~\cite{lamport1978TimeClocks, schneider1990ImplementingFaulttolerant, abraham2017BlockchainConsensus}.
For example, a simple cryptocurrency could append a list of transactions to each block.
Jointly, the confirmed blocks would form a distributed ledger.
More advanced applications could add scalability layers that only replicate selected decisions using \protChain{} while handling other state updates separately \cite{eyal2016BitcoinNGScalable,kogias2016EnhancingBitcoin,pass2018ThunderellaBlockchains}.

We require that the application offers an interface with two procedures.
\Call{getUpdate}{} returns a valid state update that \protChain{} can use for block proposals.
\Call{applyUpdate}{} passes replicated state updates to the application.
The application may have other means to access the broadcast network directly.
For example, cryptocurrencies share transactions provisionally before they are written in blocks.

\subsubsection{Cryptography} \label{sec:proto:crypto}

\protChain{} uses cryptographic hash functions for the hash-linking of blocks and the proof-of-work puzzle.
The hash function used for the linking must be cryptographically secure.
The hash function used for the proof-of-work puzzle requires the same stronger assumptions as in \nc{}~\cite{abraham2017BlockchainConsensus}.
In principle, one could separate these concerns and use two different hash functions.
For simplicity, we use a single hash function~$\hash$ satisfying both requirements.
The reader can safely assume $\hash = \operatorname{SHA3}$.

In addition, \protChain{}~uses a secure digital signature scheme~\cite[Def.~12.1, p.~442]{katz2014IntroductionModern}  given by the procedures \Call{generateKeyPair}{}, \Call{checkSignature}{}, and \Call{sign}{}.

\subsection{Protocol \protChain{}}

We start with presenting core aspects of~\protChain{} in Sections~\ref{sec:proto:firstaspect} to~\ref{sec:proto:lastaspect} and integrate them into a complete protocol in Section~\ref{sec:proto:integration}.

\subsubsection{Votes} \label{sec:proto:vote} \label{sec:proto:firstaspect}

A vote is a triple \vote{}, with $\hash\vote\leq \vthres$.
We say \varReference{} is the value voted for, \varID{} is the public key of the voting node, and \varSolution{} is the proof-of-work puzzle solution.
The threshold \vthres{} represents \protChain{}'s proof-of-work difficulty parameter and is set externally.

\subsubsection{Quorums} \label{sec:proto:quorum}

A \qsize{}-quorum is a set of~\qsize{} valid votes for the same value.
A list $Q = \{(\varID_i, \varSolution_i)\}$ represents a valid \qsize{}-quorum
for value~$\varReference$, if the following conditions hold:
\begin{enumerate}
  \item \label{quorum:size}$|Q| = \qsize$
  \item \label{quorum:threshold} $\forall\, 1 \leq i \leq \qsize \colon {\hash(\varReference, \varID_i, \varSolution_i)} \leq \vthres$
  \item \label{quorum:order} $\forall\, 1 \leq i < \qsize \colon \hash(\varReference, \varID_i, \varSolution_i) < \hash(\varReference, \varID_{i+1}, \varSolution_{i+1})$
\end{enumerate}
The first condition defines the quorum size~$k$.
The second condition ensures that all votes are valid.
The third condition eliminates duplicates and imposes a canonical order which we use for leader selection.
We write~$Q[1]$ to address the first vote in the quorum.

\begin{remark}
  The above definitions allow for single nodes providing multiple votes to a single quorum using the same public key.
  This is intentional.
  Sybil attacks are mitigated by the scarcity of votes, not by the scarcity of public keys.
\end{remark}

\subsubsection{Leader Selection} \label{sec:proto:election}

We say that node $A$ is leader for the epoch that produces $Q$ if $A$ contributed the smallest vote~$Q[1]$.
Only leaders are allowed to propose new blocks.
Nodes verify leadership based on the public key $\varID_1$, which is part of $Q[1]$.

\begin{remark} \label{rem:leader}
  Leader selection originates from the distributed system literature (e.\,g.~\cite{garciamolina1982ElectionsDistributed, castro2002PracticalByzantine, ongaro2014SearchUnderstandable, yin2019HotStuffBFT}), where it is used to improve performance in the optimistic case that the leader follows the rules.
  A similar, leader-based performance improvement has been proposed for \nc{}~\cite{eyal2016BitcoinNGScalable}.
  Our leader selection mechanism is an optimization as well.
  It reduces the number of proposals per epoch and thereby improves communication efficiency.
  Recall that the agreement protocol $\protAgree$ resolves conflicting preferences even if all nodes started with their own preferred value~(Sect.~\ref{sec:voting}).
  Thus, \protChain{} is secure even if leaders equivocate or multiple leaders are selected.
\end{remark}

\subsubsection{Blocks} \label{sec:proto:block}

A block is a proposed extension to the blockchain.
Besides the application payload, a block holds additional values that ensure orderly execution of the agreement~\protAgree{} and the leader selection mechanism according to Sections~\ref{sec:proto:vote} to~\ref{sec:proto:election}.
A valid block~$b$ contains the following information.
\begin{enumerate}
  \item $\opParent(b)$ is either the hash of a previous valid block or equal to the protocol parameter~$\genesisHash$, which characterizes the instance of the protocol;%
    \footnote{In blockchain jargon, $\genesisHash$ is the hash of the genesis block.}
  \item $\opQuorum(b)$ is a valid $k$-quorum for $\opParent(b)$;
  \item $\opPayload(b)$ is the proposed state update returned from \Call{getUpdate}{};%
  \item \sloppy $\opSignature(b)$ is a valid signature of the triple $\left(\opParent(b), \opQuorum(b), \opPayload(b)\right)$ signed with the private key corresponding to the public key in $\opQuorum(b)[1]$.
\end{enumerate}

The first condition imposes a sequential order on the list of blocks.
The second condition ensures that all nodes agree on the previous block before proposing a new block (\protAgree{}, Sect.~\ref{sec:voting}).
The forth condition restricts the ability to propose blocks to selected leaders (Sect.~\ref{sec:proto:election}).

\subsubsection{Local Block Tree} \label{sec:proto:store}

Each node locally maintains a hash-linked \emph{tree} of blocks~\blockstore.
We write~$\blockstore[h]$ to access the block~$b$ with $\hash (b) = h$.
For each block $b$, nodes maintain
\begin{enumerate}
  \item $\opHeight(b)$, the number of predecessors of~$b$ in~\blockstore,
  \item $\opState(b)$, the application state associated with $b$, and
  \item $\opVotes(b)$, the set of votes that confirm $b$.
\end{enumerate}

\subsubsection{Block Preference} \label{sec:proto:preference}

Nodes prefer block~$a$ over block~$b$ if,
\begin{enumerate}
  \item $\opHeight(a) > \opHeight(b)$, or
  \item $\opHeight(a) = \opHeight(b) \wedge |\opVotes(a)| > |\opVotes(b)|$.
\end{enumerate}
In other words, they follow the longest chain (1) between epochs and the voting protocol~\protAgree{} within each epoch (2).
This rule is ambiguous if there are multiple blocks of equal height and with the same number of confirming votes.
In this case, nodes prefer the block first received.
The embedded voting protocol~\protAgree{} makes the nodes agree on the same parent block until the end of the epoch.

\begin{remark}
Under normal operation with a constant set of nodes (i.\,e., no late joining), the longest chain rule will only be invoked to disambiguate the last epoch.
The $\varepsilon$-safety guarantee of \protAgree{} ensures that longer forks are unlikely.
\end{remark}

\subsubsection{Proof-of-Work Voting} \label{sec:proto:pow}

Nodes continuously try to find and share valid votes for their preferred block.
Recall that a valid vote $v = \left(\hash (b), \varID, s\right)$ satisfies $\hash (v) \leq \vthres$, where $b$~is the preferred block and $\varID$~is the node's public key.
Due to the properties of the hash function (Sect.~\ref{sec:proto:crypto}), the best solution strategy is iterative trial and error for different values of~$s$.
Solving this hash puzzle on physical hardware implements the stochastic clock~$P_\lambda$ presented in Section~\ref{sec:voting:model:pow} for the arrival of votes in a distributed system.
Parameter $\vthres$ must be adjusted to the desired puzzle solving rate $\lambda$ for a given technology and proof-of-work capacity.

\subsubsection{Proposing} \label{sec:proto:propose} \label{sec:proto:lastaspect}

Nodes assume leadership whenever possible.
I.\,e., they constantly check whether they can form a quorum~$Q$ where the smallest vote~$Q[1]$ is their own.
If so, they request a state update from the application, integrate it as payload into a new valid block (Sect.~\ref{sec:proto:block}), and broadcast it.

\subsubsection{Integration} \label{sec:proto:integration}

\begin{algorithm}
  \caption{Blockchain protocol \protChain{}}
  \label{alg:chain}
  \small
  \begin{algorithmic}[1]
    \Upon{init} \label{l:proto:init:0}
      \State \myID, \myKey{} $\gets$ \Call{generateKeyPair}{\,}
      \State $\opHeight(\blockstore[\initialhead]) \gets 0$
      \Comment{\initialhead: hard-coded value}
      \label{l:proto:init:n}
    \EndUpon

    \algskip

    \Upon{deliver $\mid$ vote $v$} \Comment{from line~\ref{l:sendvote}} \label{lb:rcvote}
      \State \Call{count}{$v$}
    \EndUpon

    \algskip

    \Upon{deliver $\mid$ block $b$} \Comment{from line~\ref{l:sendblock}} \label{lb:rcblock}
    \ForAll{$(\varID, \varSolution)$ \textbf{in} $\opQuorum(b)$} \Comment {count all votes}
    \State \Call{count}{$\opParent(b), \varID{}, \varSolution{}$}
    \EndFor
    \State \Call{store}{$b$}
    \EndUpon

    \algskip

    \Procedure{count}{vote $(r, p, s)$} \label{lb:countstart}
    \If {$\hash\vote \leq \vthres$}
    \Comment{vote validity, Sect.~\ref{sec:proto:vote}}
    \State $\opVotes(\blockstore[r]) \gets \opVotes(\blockstore[r]) \cup \{\text{(p, s)}\}$ \label{lb:changeA}
    \EndIf
    \EndProcedure \label{lb:countend}

    \algskip

    \Procedure{store}{block $b$} \label{lb:storestart}
    \If{$b$ is valid by Sect.~\ref{sec:proto:block} \textbf{and} $\hash(b) \not\in \blockstore$}
    \State $\opVotes(b) \gets \emptyset$
    \State $a \gets \blockstore[\opParent(b)]$
    \State $\opHeight(b) \gets \opHeight(a) + 1$
    \State $\opState(b) \gets \Call{applyUpdate}{\opState(a), \opPayload(b)}$
    \State $\blockstore[\hash(b)] \gets b$
    \EndIf
    \EndProcedure \label{lb:storeend}

    \algskip

    \Procedure{preferred}{}
    \State Find most preferred block $b$ according to Sect.~\ref{sec:proto:preference}.
    \State \Return $\hash (b)$
    \EndProcedure

    \algskip

    \Procedure{leader}{$r$}
    \State Try to form $k$-quorum $Q$ for $r$ as leader by Sect.~\ref{sec:proto:election}.
    \If{possible} \Return $Q$
    \Else\ \Return $\bot$
    \EndIf
    \EndProcedure

    \algskip

    \UponC{change to $\blockstore$} \label{lb:change}
    \State $r \gets$ \Call{preferred}{} ;\, $Q \gets$ \Call{leader}{$r$}
    \If{Q} \Comment{build block according to Sect.~\ref{sec:proto:propose}} \label{lb:proposestart}
    \State $\opParent(b) \gets r$
    \State $\opQuorum(b) \gets Q$
    \State $\opPayload(b) \gets \Call{getUpdate}{\opState(\blockstore[r])}$
    % \Comment from application
    \State $\opSignature(b) \gets \Call{sign}{b, \myKey}$
    \State \Call{store}{b}
    \State \Call{broadcast}{block b} \label{l:sendblock}
    \EndIf \label{lb:proposeend}
    \EndUponC

    \algskip

    \Procedure{proofOfWork}{} \Comment{background task} \label{lb:powstart}
    \Loop
    \State $r \gets$ \Call{preferred}{} ;\, $s \gets s+1$ \label{l:powloops}
    \If{$\hash(r, \myID, s) \leq \vthres$}
    % \Comment{vote validity, Sect.~\ref{sec:proto:vote}}
    \State $\opVotes(\blockstore[r]) \gets \opVotes(\blockstore[r]) \cup \{\text{(\myID, s)}\}$ \label{lb:changeB}
    \If{\textbf{not} \Call{leader}{$r$}} \label{l:leadershipa}
    \State \Call{broadcast}{vote $(r, \myID, s)$} \label{l:sendvote} \label{l:powloope}
       \hspace{-3cm} % I have no glue what's going on here! Remove this line if possible.
    \EndIf \EndIf \EndLoop
    \EndProcedure \label{lb:powend}
  \end{algorithmic}
\end{algorithm}

Algorithm~\ref{alg:chain} integrates Sections~\ref{sec:proto:vote} to~\ref{sec:proto:propose} into the complete protocol \protChain.
During initialization, nodes generate a key pair for the digital signature scheme~(Sect.~\ref{sec:proto:crypto}) and initialize the empty block tree (ln.~\ref{l:proto:init:0}--\ref{l:proto:init:n}).
Two event-handlers process incoming messages (ln.~\ref{lb:rcvote} and~\ref{lb:rcblock}).
Valid votes are stored (ln.~\ref{lb:countstart}--\ref{lb:countend}) and valid blocks are appended to the local block tree (ln.~\ref{lb:storestart}--\ref{lb:storeend}).
In the background, nodes continuously try to solve proof-of-work puzzles in order to cast votes for their preferred value (ln.~\ref{lb:powstart}--\ref{lb:powend}).
The path between \genesisHash{} and the preferred value in the local block tree represents the current version of the blockchain.
Whenever the local block tree changes (ln.~\ref{lb:change} triggered from ln.~\ref{lb:changeA} and \ref{lb:changeB}), nodes try to assume leadership and propose a new block (ln.~\ref{lb:proposestart}--\ref{lb:proposeend}).

We visualize a typical execution of~\protChain{} in Figure~\ref{fig:timeline} in the Appendix and compare it to sequential proof-of-work.

\begin{remark} \label{rem:c1btc}
  \protChain[1], i.\,e., \protChain{} with $k=1$, closely resembles Bitcoin as proposed by Nakamoto~\cite{nakamoto2008BitcoinPeertopeer}.
  However, we highlight one key difference:
  in Bitcoin, blocks carry a first proof-of-work confirmation of the payload proposed within the block itself.
  In~\protChain{}, the proof-of-work solutions confirm the \emph{previous} block.
  This enables parallel puzzle solving for $k>1$. %
\end{remark}

\subsection{Finality} \label{sec:proto:state}

Finality means that the application accepts a commit when it is deemed safe.
Deterministic finality is not achievable with stochastic protocols, but our concrete safety bounds help to make an informed decision on when to accept.
By implementing \protAgree{} in \protChain{}, we ensure that the commit of the state update in a block with height~$i$ is $\varepsilon$-safe as soon as a block with height $i+1$ is observed.

For example, the configuration $\alpha=\nicefrac14$, $\qsize=\k, \bar{t}=600$ from Table~\ref{tab:runtime600} is $0.0002$-safe. This implies that the worst case attacker (within the model) succeeds in causing inconsistent commits in one of 5,000 attempts.
In practice, such an attacker would find it easier to temporarily increase the share in compute power above $\alpha=\nicefrac12$, where every system solely based on proof-of-work fails.
With proof-of-work capacity being available for rent~\cite{bissias2022PricingSecurity}, this turns into an economic argument which is orthogonal to the design assumptions of \protChain.
This leads us to a brief discussion of incentives.

\subsection{Incentives} \label{sec:incentives}

It is possible to motivate participation in \protChain{} by rewarding puzzle solutions.
This requires some kind of virtual asset that %
can be transferred to a vote's public key.
Claiming the reward would depend on the corresponding private key.
\protChain{} could allot a constant reward per puzzle. %
As votes occur $\qsize$ times more frequently than blocks, \protChain's mining income would be less volatile than in \nc, making participation more attractive to risk-averse agents with small compute power.

It is tempting to demand that the reward scheme is incentive compatible, i.\,e., that correct execution is a unique maximum of the nodes' utility function.
However, it is not trivial to achieve incentive compatibility because utility of rewards \emph{outside} the system may affect the willingness to participate \emph{in} the system~\cite{ford2019RationalitySelfdefeating}.
We do not know any blockchain protocol analysis that solves this problem convincingly.
Thus, \protChain{} is designed to support rewards as a means to encourage participation, but its security intentionally does not depend on incentives.
This is a feature, not a bug.

\section{Evaluation} \label{sec:eval}

We evaluate \protChain{} by discrete event network simulation.
We implement \protChain{} and the network simulation in OCaml.
All results are reproducible with the code provided online\citerepo.

We choose the configuration $k = \k$ and $\lambda = \nicefrac{\k}{600}$, which
is optimized for $\alpha = \nicefrac14$ and $\Delta = 2''$.
Its failure probability is at most $2.2 \cdot 10^{-4}$ (see Sect.~\ref{sec:li2021}, Tab.~\ref{tab:runtime600}).
The expected block interval is 10 minutes, which enables comparison to \nc{}, more specifically Bitcoin.
For the purpose of this simulation, Bitcoin is equivalent to \protChain[1]{} with $\lambda = \nicefrac{1}{600}$ (see Sect.~\ref{sec:c1btc} Remark~\ref{rem:c1btc}).

While the worst-case propagation delay $\Delta$ is specified at design time, realistic network latencies vary.
In the simulation, we set an expected network delay $\delta$ and use it to draw individual delays for each message delivery from
\begin{enumerate}
  \item a \textbf{uniform} distribution on the interval~$[0, 2\cdot\delta]$, and
  \item an \textbf{exponential} distribution with rate $\delta^{-1}$.
\end{enumerate}

We also consider that votes may propagate faster than blocks because they are much smaller and their validation does not depend on the application state.
To this end we define
\begin{enumerate}
  \item a \textbf{simple} treatment where $\delta = \Delta = 2''$ for all messages, and
  \item a \textbf{realistic} treatment where blocks propagate with  $\delta_{\textbf{b}} = 2''$ and votes eight times faster,  $\delta_{\textbf{v}} = \nicefrac14''$.
\end{enumerate}
The cross product of the two distributions and two treatments of small messages gives us four scenarios to be simulated. %
Note that for all scenarios some delays will be greater than the assumed worst-case propagation delay $\Delta$.
For some measurements, we will raise $\delta$ beyond $\Delta$ to put the protocol under even more pressure.

{ \def\nNodes{1024}
\def\nBlocks{4096}
\def\nConfirmations{\NotUnique}
\def\nIterations{64}
\def\name{all}

 Unless stated otherwise, measurements are based on a simulated network with~\nNodes{} nodes.%
  \footnote{Measurements suggest that there are roughly 10\,000 Bitcoin nodes, while 80\,\% of the compute power is held by the top 10 agents~\cite{mariem2020AllThat}.}
  For each experiment, we average over~\nIterations{} independent executions up to block height~\nBlocks{}.
  All figures showing variation do this by plotting $\pm 1.96$~standard deviations around the mean of the \nIterations{} independent executions.  %
}
For all executions of \protChain[\k]{}, we checked for inconsistent commits, %
which did not occur. %
As another plausibility check, we verified that the simulated block intervals of \protChain[1] and \protChain[\k] match the theoretical distributions described in Section~\ref{sec:voting:config}.%

\subsection{Robustness} \label{sec:robustness}

We evaluate the robustness of \protChain[\k]{} against excessive latency, churn, and leader failure by measuring block intervals.
Recall that the simulated protocols are configured for a 10 minute interval in optimal conditions.
The puzzle solving rate is constant.
Stress implies wasted proof-of-work and hence higher observed block intervals.%
\footnote{This metric relates to the orphan rate in the literature on sequential proof-of-work.}

\subsubsection{Latency} \label{sec:eval:latency}

We use the \textbf{simple/exponential} scenario and vary the expected propagation delay $\delta$ from $\nicefrac14$ to $16$ seconds.
Recall that the choice of $k=\k$ is optimized for $\Delta = 2$ seconds.
Larger expected propagation delays put the protocol under stress.
Figure~\ref{fig:latency:simple:exponential} shows the effect of latency on the block interval.
We observe that even excessive random propagation delays ($\delta = 16$ seconds) slow down \protChain[\k]{}-consensus by only about $5\,\%$.
The \textbf{simple/uniform} scenario exhibits similar behavior.
We refrain from exploring the \textbf{realistic} treatment as it is not obvious how real network latency would affect the ratio of $\delta_{\textbf{b}}$ and $\delta_{\textbf{v}}$.

\begin{figure}
  \begin{tikzpicture}
    \def\data{sim/latency.csv}
    \begin{axis}[
      table/col sep=comma,
      width=\linewidth,
      height=0.65\linewidth,
      xmode=log,
      xtick={.25,.5,1,2,4,8,16},
      xticklabels={$\nicefrac14$,$\nicefrac12$,1,2,4,8,16},
      xlabel={expected propagation delay $\delta$},
      ylabel={block interval},
      y label style={at={(axis description cs:-0.21,.5)},rotate=0,anchor=north},
      legend cell align=left,
      legend pos= north west,
      legend style={font=\footnotesize},
      cycle list/Set1,
      cycle multi list={ Set1\nextlist mark=o\nextlist},
      ]
      \addplot table[y=mean.interval.nc-slow.exponential, x=latency] {\data};
      \addlegendentry{Bitcoin/\protChain[1]{}}
      \addplot table[y=mean.interval.proposed.exponential, x=latency] {\data};
      \addlegendentry{Proposed/\protChain[\k]{}}
      \addlegendimage{color=gray, dashed};
      \addlegendentry{$\Delta$};
      \draw[color=gray, dashed] ({rel axis cs:0,0} -| {axis cs:2,0}) -- ({rel axis cs:0,1} -| {axis cs:2,0});
      \addplot[name path=ncl, draw=none] table[y=mean.interval.low.nc-slow.exponential, x=latency] {\data};
      \addplot[name path=nch, draw=none] table[y=mean.interval.high.nc-slow.exponential, x=latency] {\data};
      \addplot[index of colormap={0 of Set1}, opacity=0.15] fill between[of=ncl and nch];
      \addplot[name path=pl, draw=none] table[y=mean.interval.low.proposed.exponential, x=latency] {\data};
      \addplot[name path=ph, draw=none] table[y=mean.interval.high.proposed.exponential, x=latency] {\data};
      \addplot[index of colormap={1 of Set1},opacity=0.15] fill between[of=pl and ph];
    \end{axis}
  \end{tikzpicture}
  \caption{%
    The effect of latency on the average block interval in the \textbf{simple/exponential} scenario. Time in seconds.
  }
  \label{fig:latency:simple:exponential}
\end{figure}
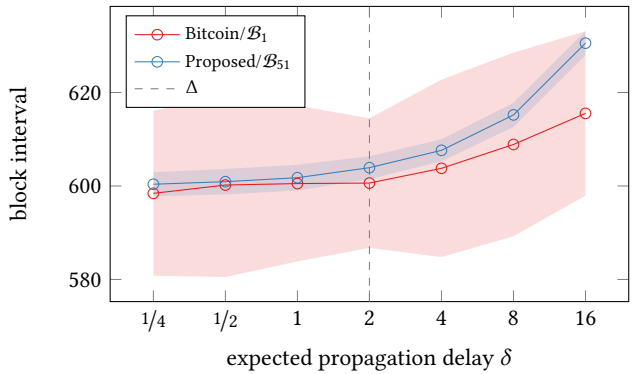

\subsubsection{Churn} \label{sec:eval:churn}

We simulate churn by muting a fraction of nodes (churn ratio) for one hour each.
Muted nodes solve proof-of-work puzzles but do not send or receive messages.
Accordingly, the votes and blocks created by muted nodes represent lost work.
We expect that the block interval is inversely proportional to the churn ratio:
if 50\,\% of the nodes are muted, the average block interval is twice as long.
Figure~\ref{fig:churn} supports this claim for both protocols.

\begin{figure}
  \begin{tikzpicture}
    \def\data{sim/churn.csv}
    \begin{axis}[
      table/col sep=comma,
      width=\linewidth,
      height=0.65\linewidth,
      xlabel={churn ratio},
      ylabel={block interval},
      y label style={at={(axis description cs:-0.21,.5)},rotate=0,anchor=north},
      legend cell align=left,
      legend pos= north west,
      legend style={font=\footnotesize},
      cycle list/Set1,
      cycle multi list={ Set1\nextlist mark=o\nextlist},
      ]
      \addplot table[y=mean.interval.nc-slow.exponential, x=churn] {\data};
      \addlegendentry{Bitcoin/\protChain[1]{}}
      \addplot table[y=mean.interval.proposed.exponential, x=churn] {\data};
      \addlegendentry{Proposed/\protChain[\k]{}}
      \addplot[name path=ncl, draw=none] table[y=mean.interval.low.nc-slow.exponential, x=churn] {\data};
      \addplot[name path=nch, draw=none] table[y=mean.interval.high.nc-slow.exponential, x=churn] {\data};
      \addplot[index of colormap={0 of Set1}, opacity=0.15] fill between[of=ncl and nch];
      \addplot[name path=pl, draw=none] table[y=mean.interval.low.proposed.exponential, x=churn] {\data};
      \addplot[name path=ph, draw=none] table[y=mean.interval.high.proposed.exponential, x=churn] {\data};
      \addplot[index of colormap={1 of Set1},opacity=0.15] fill between[of=pl and ph];
    \end{axis}
  \end{tikzpicture}
  \caption{%
    The effect of churn on the average block interval in the \textbf{realistic/exponential} scenario.
    The churn ratio describes how many of the nodes are passive at any given time.
  }
  \label{fig:churn}
\end{figure}
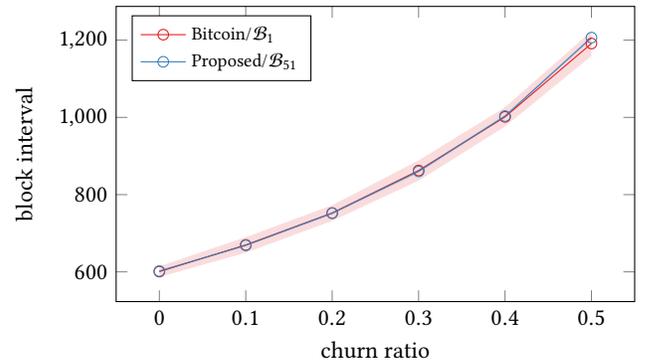

\subsubsection{Leader Failure} \label{sec:eval:failure}

\protChain{} separates proof-of-work (votes) from blocks.
Leaders selected during the epoch may fail to propose at the end of the epoch.
We model such failures by dropping block proposals randomly with constant probability (leader failure rate).

A special property of \protChain{} is that it can reuse votes for different proposals.
Honest nodes reveal at most one new vote with their proposal.
Accordingly, a lost proposal wastes at most the work of one vote.
Therefore, \protChain[\k] can recover fast.
The results in Figure~\ref{fig:failure} support this claim.
For perspective, the right end of the graph simulates a situation where an attacker can monitor all nodes' network traffic and disconnect leaders at discretion with 50\,\% success probability.
Still, the block interval grows only by about $2.5\,\%$. This effect is similar to the robustness against excessive latencies discussed in Section~\ref{sec:eval:latency}.

For \protChain[1]{}, voting, leader selection, and proposing happens in a single message.
Leader failure is similar to churn and hence has a much stronger effect (see Fig.~\ref{fig:churn}).

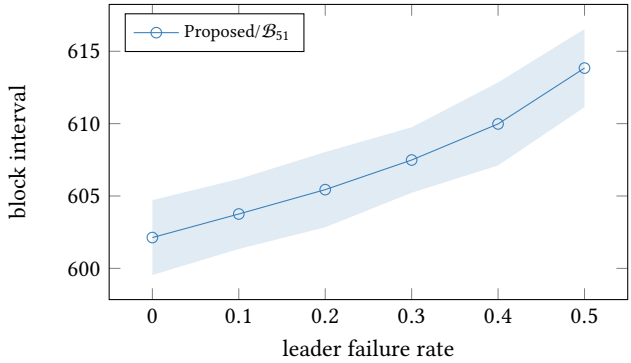
\begin{figure}
  \begin{tikzpicture}
    \def\data{sim/failure.csv}
    \begin{axis}[
      table/col sep=comma,
      width=\linewidth,
      height=0.65\linewidth,
      xlabel={leader failure rate},
      ylabel={block interval},
      y label style={at={(axis description cs:-0.21,.5)},rotate=0,anchor=north},
      legend cell align=left,
      legend pos= north west,
      legend style={font=\footnotesize},
      cycle list/Set1,
      cycle multi list={ Set1\nextlist mark=o\nextlist},
      ]
      \addplot[index of colormap={1 of Set1}, mark=o] table[y=mean.interval.proposed.exponential, x=failure] {\data};
      \addlegendentry{Proposed/\protChain[\k]{}}
      \addplot[name path=pl, draw=none] table[y=mean.interval.low.proposed.exponential, x=failure] {\data};
      \addplot[name path=ph, draw=none] table[y=mean.interval.high.proposed.exponential, x=failure] {\data};
      \addplot[index of colormap={1 of Set1},opacity=0.15] fill between[of=pl and ph];
    \end{axis}
  \end{tikzpicture}
  \caption{%
    The effect of leader failure on the block interval in the \textbf{realistic/exponential} scenario.
    The leader failure rate is the probability that a selected leader fails to propose a block.
  }
  \label{fig:failure}
\end{figure}

\subsection{Security} \label{sec:eval:security}

Zhang and Preneel~\cite{zhang2019LayCommon} propose to evaluate blockchain protocols with respect to the four security aspects
\begin{enumerate}
  \item \emph{subversion gain}, to what extent an attacker can rewrite confirmed blocks,
  \item \emph{chain quality}, how much of the confirmed blocks are proposed by the attacker,
  \item \emph{censorship susceptibility}, how long the attacker can block certain transactions, and
  \item \emph{incentive compatibility}, how much reward the attacker can collect by deviating from the protocol.
\end{enumerate}

Our approach is to derive subversion gain from the $\varepsilon$-safety of~\protAgree{} and then evaluate chain quality and censorship susceptibility jointly.
This is sufficient because both aspects depend on the attacker being selected as a leader.
Turning to incentive compatibility, we argue in Section~\ref{sec:incentives} why it seems impossible to prove this for realistic utility functions.
Zhang and Preneel use a restricted notion in which the attacker utility is the share of rewards assigned by the protocol.
We can evaluate their definition of incentive compatibility, along with chain quality and censoring.

\subsubsection{Subversion Gain}
We provide a consistency analysis for the agreement~\protAgree{} in Section~\ref{sec:voting}.
The proposed \protChain[\k]{} executes \protAgree[\k]{} for each appended block.
The probability that an $\alpha = \nicefrac14$ attacker in a $\Delta=2''$ synchronous network succeeds in causing inconsistent state updates (e.\,g., double spend) is \num{2.2e-4}  (see Tab.~\ref{tab:runtime600}).
The proposed protocol meets this guarantee after one block confirmation, i.\,e., after about 10 minutes.%
We assume that applications wait for this confirmation and conclude that subversion gain is not a practical concern for \protChain[\k].

\begin{figure}
  \begin{tikzpicture}
    \def\data{sim/attacker.csv}
    \begin{axis}[
      table/col sep=comma,
      width=\linewidth,
      height=0.65\linewidth,
      xlabel={$\alpha$},
      ylabel={share},
      y label style={at={(axis description cs:-0.21,.5)},rotate=0,anchor=north},
      legend cell align=left,
      legend pos= north west,
      legend style={font=\footnotesize},
      cycle list/Set1,
      cycle multi list={ Set1\nextlist mark=o\nextlist},
      ]
      \addplot table[y=share.blocks.realistic-exponential, x=attacker] {\data};
      \addlegendentry{blocks};
      \addplot table[y=share.votes.realistic-exponential, x=attacker] {\data};
      \addlegendentry{votes};
      \addplot[gray] table[y=attacker, x=attacker] {\data};
      \addlegendentry{honest ($=\alpha$)};
      \addplot[name path=bl, draw=none] table[y=share.blocks.low.realistic-exponential, x=attacker] {\data};
      \addplot[name path=bh, draw=none] table[y=share.blocks.high.realistic-exponential, x=attacker] {\data};
      \addplot[index of colormap={0 of Set1},opacity=0.15] fill between[of=bl and bh];
      \addplot[name path=vl, draw=none] table[y=share.votes.low.realistic-exponential, x=attacker] {\data};
      \addplot[name path=vh, draw=none] table[y=share.votes.high.realistic-exponential, x=attacker] {\data};
      \addplot[index of colormap={1 of Set1}, opacity=0.15] fill between[of=vl and vh];
      \addplot[index of colormap={0 of Set1}, dashed] table[y=share_of_blocks, x=alpha] {hotpow-withholding-mc.csv};
      \addplot[index of colormap={1 of Set1}, dashed] table[y=share_of_votes, x=alpha] {hotpow-withholding-mc.csv};
    \end{axis}
  \end{tikzpicture}
  \caption{%
    The attacker's share of confirmed blocks and votes as functions of $\alpha$ in the \textbf{realistic/exponential} scenario.
    The attacker uses vote withholding against \protChain[\k]{}.
    The gray line shows the expected shares without attack.
    The dashed lines show results from the Markov chain model.
    The solid lines show the validation in the network simulator.
  }
  \label{fig:attacker}
\end{figure}
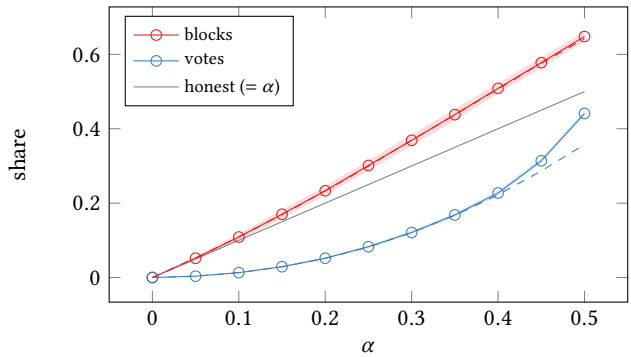

\subsubsection{Chain Quality, Censoring, and Incentives} \label{sec:eval:security:censoring}
Chain quality measures the share of confirmed blocks proposed by the attacker.
Censoring is possible only if the attacker controls the proposed block payload.
Thus, chain quality and censoring reduce to the question of how often an attacker can take leadership. %

A common weakness of sequential proof-of-work protocols relates to information withholding.
Block withholding, proposed by Eyal and Sirer~\cite{eyal2014MajorityNot}, enables selfish mining against Bitcoin. %
\protChain{} is not vulnerable to block withholding because selected leaders who do not propose a block are quickly replaced (see Sect.~\ref{sec:eval:failure}).
The remaining information to be considered in withholding attacks are votes (see Sect.~\ref{sec:voting:alpha} for \protAgree{}; related~\cite{bissias2020BobtailImproved}).
In \protChain{}, the attacker can prolong an epoch by withholding votes until the honest nodes can form a $k$-quorum themselves.
The attacker can use the additional time to mine the smallest vote and be selected as leader.

We first analyze the effectiveness of vote withholding in a single epoch using a Markov chain model (see Appendix~\ref{apx:mc_censor}).
Then we use the network simulator to confirm the results for executions of the protocol over multiple epochs.

Figure~\ref{fig:attacker} shows the success rate of the attacker in red and his number of confirmed votes in blue. %
Solid lines originate from the network simulator and dashed lines from the Markov chain model.
Both evaluation methods concur in the main result: a withholding attacker can become the leader in about $1.3 \cdot \alpha$ cases (65\,\% for $\alpha=50\%$).
His advantage in taking leadership comes at the price of fewer confirmed votes.
If rewards are proportional to votes, this tells us that vote withholding is disincentivized.
For comparison with \nc{}, block withholding strategies give an $\alpha = \nicefrac13$ attacker an advantage of $1.5 \cdot \alpha$.
This factor raises to $2 \cdot \alpha$ for $\alpha=\nicefrac12$~\cite{sapirshtein2016OptimalSelfish}.
Moreover, successful selfish miners receive more rewards than without attack.
The results indicate that \protChain[\k]{} offers higher chain quality, is less susceptible to censorship, and offers fewer incentives to attack than \nc.

\subsection{Overhead} \label{sec:eval:overhead}   %

\begin{figure}
  \begin{tikzpicture}
    \def\data{sim/size.csv}
    \begin{axis}[
      table/col sep=comma,
      width=\linewidth,
      height=0.65\linewidth,
      xmode=log,
      xtick={2,8,32,128,512,2048,8192},
      xticklabels={2,8,32,128,512,2048,8192},
      xlabel={number of nodes},
      ylabel={\# messages / \# blocks / $k$},
      y label style={at={(axis description cs:-0.21,.5)},rotate=0,anchor=north},
      legend cell align=left,
      legend style={font=\footnotesize, anchor=east, at={(0.97,0.72)}},
      cycle list/Set1,
      cycle multi list={ Set1\nextlist mark=o\nextlist},
      ]
      \addplot table[y=messages.rel.nc-slow.exponential, x=size] {\data};
      \addlegendentry{Bitcoin/\protChain[1]{}};
      \addplot table[y=messages.rel.proposed.exponential, x=size] {\data};
      \addlegendentry{Proposed/\protChain[\k]{}};
      \label{curve:size:proposed}
      \addplot[name path=ncl, draw=none] table[y=messages.rel.low.nc-slow.exponential, x=size] {\data};
      \addplot[name path=nch, draw=none] table[y=messages.rel.high.nc-slow.exponential, x=size] {\data};
      \addplot[index of colormap={0 of Set1}, opacity=0.15] fill between[of=ncl and nch];
      \addplot[name path=pl, draw=none] table[y=messages.rel.low.proposed.exponential, x=size] {\data};
      \addplot[name path=ph, draw=none] table[y=messages.rel.high.proposed.exponential, x=size] {\data};
      \addplot[index of colormap={1 of Set1},opacity=0.15] fill between[of=pl and ph];
    \end{axis}
  \end{tikzpicture}
  \caption{%
    Number of broadcast messages per block divided by $k$ for networks of different size in the \textbf{realistic/exponential} scenario.
  }
  \label{fig:size:real:exponential}
\end{figure}
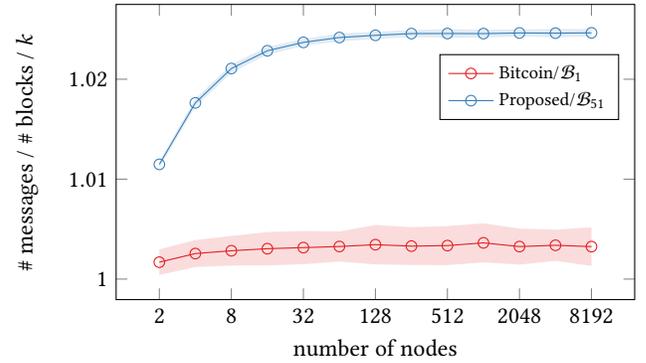

\nc{} requires at least one message broadcast per appended block, namely the block itself, independent of the number of participating nodes.
\protChain{} adds~$\qsize$ message broadcasts per block---one for each vote.
We evaluate the actual number of sent messages in the network simulator.
Figure~\ref{fig:size:real:exponential} shows the number of broadcast messages as a function of the number of blocks and $k$.
Observe that \protChain[\k]{} plateaus at about $1.025 \cdot \qsize$, i.\,e., $52$ broadcasts per block.
This number remains stable as the network scales up.

While the constant factor $\qsize$ may matter for practical networks, it is worth pointing out that vote messages are much smaller than blocks.
Under the conservative assumptions of 256 bits each for the block reference and the public key, and 64 bits for the puzzle solution, a vote is as small as 72\,B.%
\footnote{Bitcoin shortens public keys to 160 bits and uses solutions of 32
bits.
Its blocks are in the order of 1\,MB.}%

The votes also cause a constant storage overhead.
\protChain{} persists the complete quorum of $\qsize$ votes for future verification.
Note that the reference $\varReference$ is redundant in all votes and needs to be stored only once.
Hence, under the assumptions leading to 72\,B message size, the storage overhead of \protChain[\k]{} is about 2\,kB per block.
This is less than 0.2\,\% of Bitcoin's average block size in April 2022. %

\subsection{Detecting Network Splits} \label{sec:eval:split}

The assumption of a $\Delta$-synchronous network is unavoidable for proof-of-work protocols:
delaying the propagation of a defender's puzzle solution is equivalent to reducing his compute power.
With unbounded delays, even a weak attacker can solve sufficiently many puzzles before the defender's solutions propagate~\cite{pass2017AnalysisBlockchain}.

While network splits clearly violate this assumption, we still want to highlight that \protChain{} allows for faster detection of such events than \nc{}.
In \protChain{}, each vote carries one puzzle solution.
The activation delay is exponentially distributed with rate~$\lambda$ (see Sect.~\ref{sec:voting:model:pow}).
In an intact network, the time between received votes should follow the same distribution.
This allows nodes to test the hypothesis of being eclipsed.
For \protChain[\k]{}, a node can distinguish a network split from normal conditions with high confidence after \num{82} seconds of not receiving a vote (error probability $p = 0.1\,\%$).
For comparison, the same hypothesis test would require more than an hour of observation in Bitcoin.

\section{Discussion} \label{sec:discussion}

We discuss our contributions from several perspectives.
Section~\ref{sec:related:analyses} compares the security analysis of \protAgree{} to the relevant literature.
Section~\ref{sec:related:protocols} positions the family of protocols \protChain{} in the design space of blockchain protocols.
Limitations and directions for future work are discussed in Section~\ref{sec:future}.

\subsection{Related Security Analyses} \label{sec:related:analyses}

\begin{table}
  \caption{Comparison of related security analyses.}
  \label{tab:blockchainsecurity}
  \bgroup
  \def\cm{\checkmark}
  \newcommand{\rh}[1]{\rotatebox{45}{#1}}
  \newcolumntype{Y}{>{\centering\arraybackslash}X}
  \begin{tabularx}{\linewidth}{lYYYYYYY@{\hspace{3em}}}
    \toprule
    & \rh{Garay et al., 2015~\cite{garay2015BitcoinBackbone}}
    & \rh{Pass et al., 2017~\cite{pass2017AnalysisBlockchain}}
    & \rh{Kiffer et al., 2018~\cite{kiffer2018BetterMethod}}
    & \rh{Ga\v zi et al., 2020~\cite{gazi2020TightConsistency}}
    & \rh{Dembo et al., 2020~\cite{dembo2020EverythingRace}}
    & \rh{Li et al., 2021~\cite{li2021CloseLatency}}
    & \rh{this paper} \\
    \midrule
    time \\
    \quad - discrete slots & \cm & \cm & \cm & \cm \\
    \quad - continuous & & & & & \cm & \cm & \cm \\
    \midrule
    synchrony & slot & $\Delta$ & $\Delta$ & $\Delta$ & $\Delta$ & $\Delta$ & $\Delta$ \\
    \midrule
    security \\
    \quad - eventual & \cm & \cm & \cm & \cm & \cm \\
    \quad - $\varepsilon$-bounded & & & & & & \cm & \cm \\
    \midrule
    Markov chains & & & \cm & & & & \cm \\
    \bottomrule
  \end{tabularx}
  \egroup
\end{table}

Our security analysis of~\protAgree{} is inspired by the literature on Bitcoin security.
Table~\ref{tab:blockchainsecurity} summarizes selected landmark contributions.

The first formal security argument of the so-called ``Bitcoin backbone protocol''~\cite{garay2015BitcoinBackbone} discretized time in slots.
Each slot represents one puzzle trial attempt.
Communication is synchronous, i.\,e., messages are delivered at the end of each slot.
Security proofs for consistency and chain quality were given asymptotically in the number of slots.
The work formally established the eventual consistency of \nc{} in synchronous networks.

Follow-up work generalized the main results of~\cite{garay2015BitcoinBackbone} for a $\Delta$-synchronous communication model.
It allows messages to be delivered in future slots~\cite{pass2017AnalysisBlockchain}.
Further refinements using Markov chain models resulted in tighter, but still asymptotic bounds~\cite{kiffer2018BetterMethod}.
Recently, two research groups independently derived optimal bounds~\cite{dembo2020EverythingRace,gazi2020TightConsistency} for the attacker threshold $\alpha$.
Dembo et al.~\cite{dembo2020EverythingRace} use continuous time and model proof-of-work as a Poisson process.
All analyses cited above use asymptotic security notions.
A recent contribution to AFT\,'21 breaks with this tradition and provides concrete failure bounds for \nc{} after waiting for a given confirmation time~\cite{li2021CloseLatency} (comp.\ Sect.~\ref{sec:li2021}).
Concurrent work improves these bounds~\cite{gazi2021PracticalSettlement} and simplifies the analysis~\cite{guo2022BitcoinLatency}.

Our analysis of \protAgree{} establishes $\varepsilon$-safety in $\Delta$-synchronous networks.
We use Poisson processes to model proof-of-work in continuous time, and Markov chains as an analytic tool.

\subsection{Related Protocols} \label{sec:related:protocols}

We do not attempt to provide a complete map of the design space for replication protocols since other researchers have specialized on this
task~\cite{bonneau2015SoKResearch, cachin2017BlockchainConsensus, bano2019SoKConsensus, garay2020SoKConsensus}.
Instead, we compare \protChain{} to some of its closest relatives along selected dimensions.

Research on agreement and state replication began in the late 1970s, initially only considering benign but unreliable behavior~\cite{lamport1978TimeClocks, pease1980ReachingAgreement}.
In the early 1980s, Lamport extended the discussion to adversarial behavior. 
He coined the term Byzantine Fault Tolerance (BFT)~\cite{lamport1982ByzantineGenerals}, impliying that at most $f$ out of $n = 3f + 1$ identified nodes may deviate from the protocol.
Early BFT protocols~\cite[e.\,g.,][]{castro2002PracticalByzantine, dwork1988ConsensusPresence} rely on communication between all $n$ nodes for each value agreed upon.
This results in quadratic communication complexity and renders the protocols impractical for more than a dozen nodes.

In that light, Bitcoin~\cite{nakamoto2008BitcoinPeertopeer} can be seen as a technical breakthrough~\cite{abraham2017BlockchainConsensus, narayanan2017BitcoinAcademic}.
Nodes can join the network without obtaining permission from a central gatekeeper.
Active participation, i.\,e., proposing new state updates, is limited by the nodes' ability to solve proof-of-work puzzles.
This yields sub-quadratic communication complexity and allows Bitcoin to scale to thousands of nodes.

However, the sequential proof-of-work scheme proposed with Bitcoin faces a fundamental trade-off between security and throughput.
It supports only one state update per block, while the time between blocks must be in the order of seconds to minutes~\cite{li2021CloseLatency}.
Bitcoin-NG~\cite{eyal2016BitcoinNGScalable} tries to resolve this conflict by separating leader selection and transaction processing.
The miner of a block becomes responsible for appending multiple consecutive state updates until the next leader emerges with the next mined block.
To discipline the leader, Bitcoin-NG relies on incentives.

A number of protocols extend on the idea and try to shift trust from a single to multiple leaders.
E.\,g., Byzcoin~\cite{kogias2016EnhancingBitcoin} selects a committee from the last successful miners, who then run a conventional BFT protocol to agree on the latest state.
Similar layered protocols evolved concurrently~\cite[][]{decker2016BitcoinMeets} and afterwards~\cite{pass2017HybridConsensus, abraham2017SolidaBlockchain, pass2018ThunderellaBlockchains}.
However, the synchronization between the different consensus layers increases protocol complexity~\cite{pass2017HybridConsensus, pass2018ThunderellaBlockchains, abraham2017SolidaBlockchain} and is a source of concern~\cite{decker2016BitcoinMeets, kogias2016EnhancingBitcoin}.
Moreover, all layered protocols assume that the attacker cannot corrupt committee members selectively.

\protChain{} does not layer different consensus mechanisms.
It implements state replication directly from proof-of-work and broadcast primitives.
As demonstrated in Section~\ref{sec:eval:failure}, \protChain{} tolerates selective corruption of committee members.

Other protocols with non-sequential proof-of-work have been proposed.
E.\,g., Phantom~\cite{sompolinsky2021PhantomGhostdag} replaces the linear blockchain data structure with a directed acyclic graph (DAG) for better scalability and faster first confirmation in latent networks.
However, it also increases protocol complexity;
in particular, deriving a total order of state updates from the DAG is NP-hard.
A related idea is to operate multiple sequential blockchains in parallel~\cite{yu2020OHIEBlockchain, bagaria2019PrismDeconstructing, fitzi2020LedgerCombiners, kiffer2018BetterMethod}. %
A close relative of \protChain{} is Bobtail~\cite{bissias2020BobtailImproved}.
It uses multiple proof-of-work puzzles per block, like \protChain{}, but binds a preliminary state update to each vote.
Votes reference earlier votes, hence Bobtail mixes elements of parallel and sequential proof-of-work.
All cited non-sequential protocols lack a principled analysis on when to accept state updates.
\protChain{} is the first non-sequential proof-of-work protocol that comes with concrete bounds for safety.

A separate line of research explores alternatives to the wasteful proof-of-work primitive.
Proof-of-stake protocols~\cite[e.\,g.,][]{chen2019AlgorandSecure, daian2019SnowWhite, david2018OuroborosPraos} select committee members for participation in an agreement procedure based on the distribution of stake in the system's cryptocurrency.
This creates a cyclic dependency which, we think, has not been solved convincingly.
Other proof-of-x protocols try to provide useful services like file storage~\cite{benet2017ProofReplication} or radio network coverage~\cite{haleem2018HeliumDecentralized}.
However, these protocols seem to rely on heuristics for truthful resource accounting.

In the meantime, research on permissioned protocols caught up.
Current BFT-style protocols achieve sub-quadratic communication complexity, enabling deployments with hundreds of nodes~\cite{yin2019HotStuffBFT, abraham2019CommunicationComplexity}.
These protocols depend on identities and can scale applications as long as the BFT assumptions hold.

For digital currencies, it might be worth abandoning state machine replication completely and implement digital asset transfer directly from reliable %
broadcast~\cite{guerraoui2019ConsensusNumber}.
Promising proposals exist for the BFT assumptions~\cite{baudet2020FastPayHighperformance}, proof-of-work~\cite{sompolinsky2016SPECTREFast}, and proof-of-stake~\cite{sliwinski2021AsynchronousProofofstake}.
However, this approach restricts the versatility of the application layer.
It cannot support arbitrary smart contract logic.

\subsection{Limitations and Future Work} \label{sec:future}

We have presented a permissionless replication protocol that achieves $\varepsilon$-safety in $\Delta$-synchronous networks with computationally bounded attackers.
Although our model is widely accepted in the literature~\cite{garay2015BitcoinBackbone, pass2017AnalysisBlockchain}, it is worth discussing its assumptions.

We assume a fixed puzzle solving rate $\lambda$, but in practice agents can add and remove compute power at their discretion.  %
Practical systems try to stabilize $\lambda$ with a control loop known as difficulty adjustment algorithm (DAA)~\cite{%
  kraft2016DifficultyControl,%
  fullmer2018AnalysisDifficulty,%
  hovland2017NonlinearFeedback,%
  harding2020RealtimeBlock,%
  bissias2020RadiumImproving%
}.
For \protChain{}, the accuracy of a DAA could increase in $\qsize$ as every additional vote provides a data point for the estimation of $\lambda$.
Turning to the synchronous network assumption, as shown in Section~\ref{sec:eval:split}, the response time to detect network splits decreases for larger $\qsize$.
This relates to the CAP theorem~\cite{gilbert2002BrewerConjecture}, which states that every distributed system has to sacrifice one out of consistency, availability, and partition tolerance.
\protChain{}, as specified in Algorithm~\ref{alg:chain}, favors availability over consistency.
The trade-off could be changed in favor of consistency by implementing the split detection.
Such a variant of the protocol could notify the application layer to withhold commits and trigger out-of-band resolutions.

The perhaps most problematic assumption is that the attacker's share in compute power $\alpha$ is small (see Table~\ref{tab:optima}).
Violations, especially $\alpha > \nicefrac12$, are catastrophic, but have been observed in practice~\cite{cryptoslate2019percent51attacks}.
Note that the theory in Section~\ref{sec:voting:alpha} does apply for values of $\alpha>\nicefrac13$, but the resulting failure probabilities~$\varepsilon$ are unattractive.
This contrasts with the BFT literature, which requires a hard upper threshold of $n = 3f + 1$ to satisfy an absolute notion of safety.

This leads us to future work.
Our evaluation of \protChain{} is limited to one instance, \protChain[\k], and uses independent propagation delays on a fully connected graph.
We chose the instance for comparability with Bitcoin.
Tests of other protocol configurations with more realistic topologies %
could complete the picture.
However, as the literature reports discrepancies between the topology implied at design time and the one observed in practice~\cite{delgadosegura2019TxProbeDiscovering, mariem2020AllThat}, it is not obvious what topology would be appropriate.
A different direction is to explore improvements in the optimistic case by including %
application-level payloads into \protChain's vote messages.
E.\,g., one could add transactions that do not require consensus~\cite{guerraoui2019ConsensusNumber, sliwinski2021AsynchronousProofofstake, baudet2020FastPayHighperformance} or implement a staging of state updates~\cite{zhang2022NCmaxBreaking}. %
Finally, as explained in~Section~\ref{sec:incentives}, we refrain from designing an incentive mechanism for \protChain{}.
Possible approaches are to search reward-optimizing strategies using Markov Decision Processes (MDPs)~\cite{sapirshtein2016OptimalSelfish, zhang2019LayCommon, zur2020EfficientMDP} or reinforcement learning~\cite{hou2021SquirRLAutomating}.

\section{Conclusion} \label{sec:conclusion}

The proposed family of protocols \protAgree{} proves that unidentified nodes can reach agreement with guaranteed liveness and $\varepsilon$-safety in a $\Delta$-synchronous network using proof-of-work.
The family of protocols \protChain{} shows that parallel proof-of-work enables blockchain protocols with concrete security bounds. %
With \qsize{} chosen as described, \protChain{} enables permissionless state replication that serves certain applications better than existing systems.

It is worth noting that proof-of-work is a wasteful way of establishing agreement.
Many alternatives exist if nodes are identifiable.
The value of this research is to get better guarantees from protocols when there is no alternative to proof-of-work.

\bibliographystyle{ACM-Reference-Format}
\bibliography{code, zotero}

\appendix
\def\apxsection{\section}

\balance
\apxsection{Markov Chain Model for Chain Quality and Censoring}
\label{apx:mc_censor}

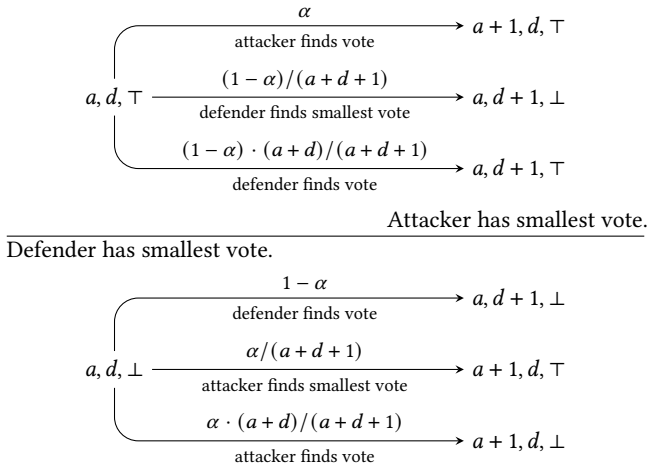
\begin{figure}
  \centering %
  \begin{tikzpicture}[>=stealth,x=17em,y=-7ex]
    \tikzstyle{state}=[]
    \tikzstyle{transition}=[->, draw, rounded corners=1em]
    \tikzstyle{prob}=[anchor=south, font=\small, xshift=-0.5em]
    \tikzstyle{intuition}=[anchor=north, font=\footnotesize, xshift=-0.5em]
    \node[state] (top) at (0,1) {$a,d,\top$};
    \node[state] (ta) at (1,0) {$a+1,d,\top$};
    \node[state] (tb) at (1,1) {$a,d+1,\bot$};
    \node[state] (tc) at (1,2) {$a,d+1,\top$};
    \path[transition] (top) |- (ta);
    \path[transition] (top) -- (tb);
    \path[transition] (top) |- (tc);
    \node[prob] at (0.5, 0) {$\alpha$};
    \node[prob] at (0.5, 1) {$(1 - \alpha)/(a + d + 1)$};
    \node[prob] at (0.5, 2) {$(1 - \alpha)\cdot(a+d)/(a + d + 1)$};
    \node[intuition] at (0.5, 0) {attacker finds vote};
    \node[intuition] at (0.5, 1) {defender finds smallest vote};
    \node[intuition] at (0.5, 2) {defender finds vote};
  \end{tikzpicture}

  \strut\hfill Attacker has smallest vote.
  \hrule
  Defender has smallest vote. \hfill\strut

  \begin{tikzpicture}[>=stealth,x=17em,y=-7ex]
    \tikzstyle{state}=[]
    \tikzstyle{transition}=[->, draw, rounded corners=1em]
    \tikzstyle{prob}=[anchor=south, font=\small, xshift=-0.5em]
    \tikzstyle{intuition}=[anchor=north, font=\footnotesize, xshift=-0.5em]
    \node[state] (top) at (0,1) {$a,d,\bot$};
    \node[state] (ta) at (1,0) {$a,d+1,\bot$};
    \node[state] (tb) at (1,1) {$a+1,d,\top$};
    \node[state] (tc) at (1,2) {$a+1,d,\bot$};
    \path[transition] (top) |- (ta);
    \path[transition] (top) -- (tb);
    \path[transition] (top) |- (tc);
    \node[prob] at (0.5, 0) {$1 - \alpha$};
    \node[prob] at (0.5, 1) {$\alpha/(a + d + 1)$};
    \node[prob] at (0.5, 2) {$\alpha\cdot(a+d)/(a + d + 1)$};
    \node[intuition] at (0.5, 0) {defender finds vote};
    \node[intuition] at (0.5, 1) {attacker finds smallest vote};
    \node[intuition] at (0.5, 2) {attacker finds vote};
  \end{tikzpicture}
  \caption{Markov chain of vote withholding to gain leadership in an epoch of \protChain. Smallest vote implies leadership.}
  \label{fig:mcmc}
\end{figure}

\bgroup

\newcommand{\success}{\textsc{success}}
\newcommand{\fail}{\textsc{fail}}

We describe the Markov chain model used in Section~\ref{sec:eval:security:censoring}.
Let the triple $(a, d, l)$ be the current Markov state,
where $a \in \mathbb N$ denotes the number of withheld attacker votes,
$d \in \mathbb N$ denotes the number of votes found by the defender,
and $l \in \{\bot, \top\}$ is true if the attacker currently holds the smallest vote.
The initial state is $(1, 0, \top)$ with probability $\alpha$ and $(0, 1, \bot)$ otherwise.

Figure~\ref{fig:mcmc} depicts the probabilistic state transitions.
As in Section~\ref{sec:voting:alpha}, transitions occur at the ticks of the stochastic clock.
With probability $\alpha$, the attacker finds a new vote.
The probability that it is the smallest (leading) vote is $1/(a + d +1)$.
This expression follows from $a+d$ old votes cutting the domain into $a+d+1$ bins.
As the hash function's outputs are indistinguishable from a uniform distribution, the expected bin size is $1/(a+d+1)$.
To simplify the figure, we do not show the two terminal states \success{} and \fail.
The former is reached when the attacker proposes a valid block ($l \wedge a + d \geq \qsize$).
Conversely, if $\neg l \wedge d \geq \qsize$, the defenders propose a block.
In all other cases, the epoch continues.

For $k=\k$, the resulting Markov chain has $5204$
states.
We evaluate it with Monte Carlo simulation for
1\,000\,000 epochs, $k=\k$ and $\alpha$ in the range $\left[0, \nicefrac12\right]$.
To validate these results in the context of the protocol and network latency, we implement the same attack in the network simulator and collect data from $\nIterations$ independent executions by $\nNodes$ nodes up to block height $\nBlocks$.
In both cases, we measure chain quality and censorship susceptibility by counting terminations in the state \success{}.
In addition, we analyze incentive compatibility by counting attacker votes.

\egroup %

\bgroup
\newsavebox{\timelineCheckedList}
\savebox{\timelineCheckedList}{\tikz{
    \node[] at (0,0) {\faList};
    \node[draw, circle, fill=white, inner sep=0.5pt] at (3pt,-3pt) {\tiny\faCheck};
}}
\newcommand{\faCheckedList}{\usebox{\timelineCheckedList}}

\tikzset{puzzle/.style={draw, rectangle, fill=white, inner sep=4pt, label={center:\resizebox{5pt}{!}{\faCheck}}}}
\tikzset{leadingPuzzle/.style={draw, rectangle, fill=white, inner sep=5pt, label={center:\resizebox{6pt}{!}{\faTrophy}}}}
\tikzset{block/.style={draw, rectangle, fill=white, inner sep=10pt, label={center:\faList}}}
\tikzset{btcBlock/.style={draw, rectangle, fill=white, inner sep=10pt, label={center:\faCheckedList}}}
\tikzset{boundTo/.style={->, draw, color=gray}}
\tikzset{node/.style={draw, dotted, color=gray}}
\tikzset{epoch/.style={draw, dashed}}

\newcommand{\legend}{
  \begin{tikzpicture}[yscale=-0.66, xscale=0.66]
    \path (0,0) node[] {\faCheck} -- +(1,0) node[anchor=west] {puzzle solution\strut};
    \path (0,1) node[] {\faTrophy} -- +(1,0) node[anchor=west] {leading puzzle solution\strut};
    \path (0,2) node[] {\faList} -- +(1,0) node[anchor=west] {state update\strut};

    \begin{scope}[shift={(+8,-3)}]

      \path (0,3) node[draw,rectangle, inner sep=1.5ex] {} -- +(1,0) node[anchor=west] {message\strut};

      \path[boundTo, <-] (-0.35,4) -- +(0.7,0) {};
      \node[anchor=west] at (1,4) {hash-reference\strut};

      \path[epoch] (-.25, 4.60) -- +(0,0.80) {};
      \path[epoch] ( .25 ,4.60) -- +(0,0.80) {};
      \node[anchor=west] at (1,5) {epoch\strut};
    \end{scope}
  \end{tikzpicture}
}

\input{fig/timeline.tex}

\begin{figure*}
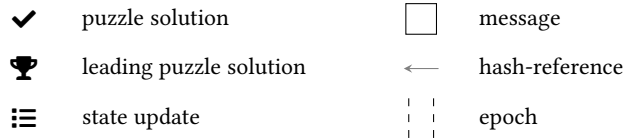
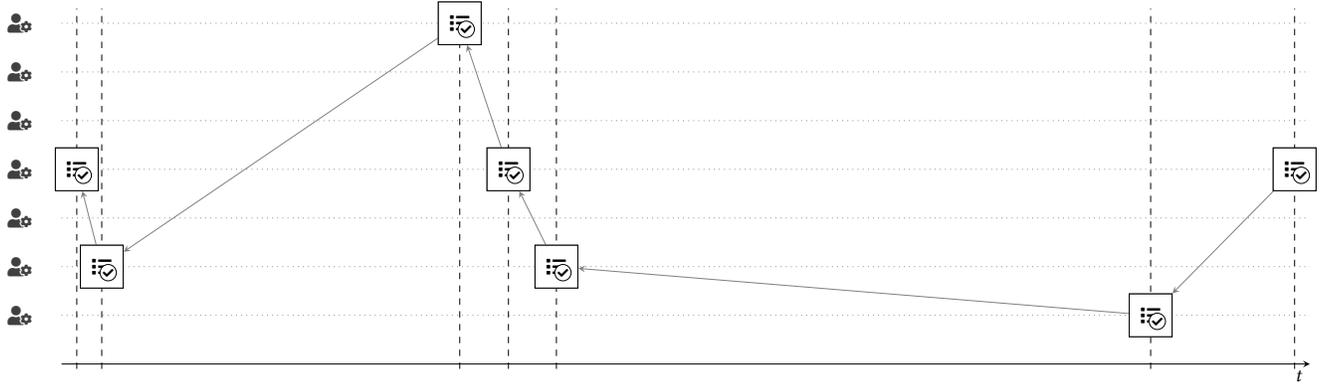
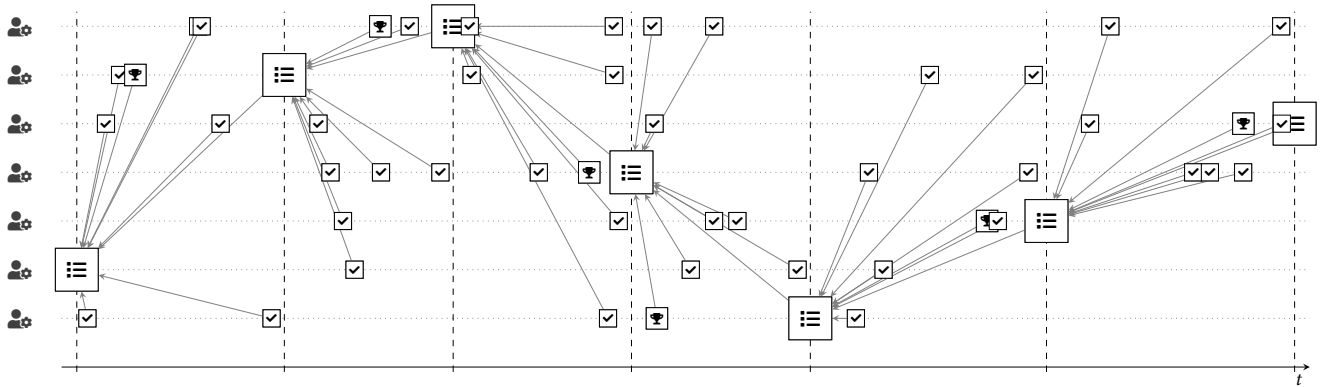

  \begin{subfigure}[b]{\linewidth}
    \vspace{3ex}
    \centering
    \legend
    \caption{Legend.}
    \label{fig:timeline_legend}
  \end{subfigure}
  \begin{subfigure}[b]{\linewidth}
    \resizebox{\linewidth}{!}{ \bitcoind }
    \caption{
      In Bitcoin, puzzles are solved sequentially.
      Puzzle solving times are exponentially distributed, implying irregular block intervals.
    }
    \label{fig:timeline_bitcoin}
  \end{subfigure}
  \begin{subfigure}[b]{\linewidth}
    \vspace{3ex}
    \resizebox{\linewidth}{!}{ \bkl }
    \caption{
      In \protChain{}, multiple smaller puzzles are solved in parallel for each block.
      One of the $\qsize=8$ votes is chosen as leader and the corresponding miner proposes the next block (see Sec.~\ref{sec:proto:election}).
      The embedded agreement protocol~\protAgree{} guarantees $\varepsilon$-safety for the previous block (see Sec.~\ref{sec:voting}).
      Observe that \protChain{} has more regular block intervals and enables more frequent rewards for miners.
    }
    \label{fig:timeline_hotpow}
  \end{subfigure}
  \caption{
    Simulated executions of Bitcoin and \protChain[8]{} on $n=7$ nodes ($y$-axis) over time ($x$-axis).
  }
  \label{fig:timeline}
\end{figure*}

\egroup %

\end{document}